\documentclass[10pt,a4paper]{article}

\usepackage{amsfonts}
\usepackage{amsmath}
\usepackage{amssymb}
\usepackage{amsthm}
\usepackage{bussproofs}
\usepackage{enumerate}
\usepackage{mdwlist}
\usepackage{mathrsfs}
\usepackage{tabularx}

\usepackage{fullpage}

\input xy
\xyoption{all}

\newtheorem{theorem}{Theorem}[section]
\newtheorem{proposition}[theorem]{Proposition}

\newtheorem{corollary}[theorem]{Corollary}

\theoremstyle{remark}
\newtheorem{definition}[theorem]{Definition}
\newtheorem{remark}[theorem]{Remark}
\newtheorem{example}[theorem]{Example}

\def\NN{\mathbb{N}}

\def\BB{\mathbb{B}}

\newcommand{\pair}[1]{\langle #1 \rangle}

\newcommand{\initSeg}[2]{[#1](#2)}

\newcommand{\ext}[1]{\widehat{#1}}

\newcommand{\Eloise}{\exists\mbox{loise}}
\newcommand{\Abelard}{\forall\mbox{belard}}

\newcommand{\ps}[1]{\bar{#1}}
\newcommand{\wideps}[1]{\overline{#1}}
\newcommand{\one}{{\bf 1}}

\newcommand{\modcont}{\mathscr{C}^\omega}

\newcommand{\T}{\mathcal{T}}
\newcommand{\N}{\NN}
\newcommand{\Bool}{\BB}
\newcommand{\B}{\BB}
\newcommand{\F}{\mathsf{F}}

\newcommand{\HAomega}{{\sf HA}^\omega}
\newcommand{\PAomega}{{\sf PA}^\omega}
\newcommand{\EHAomega}{{\sf E\mbox{-}HA}^\omega}
\newcommand{\EPAomega}{{\sf E\mbox{-}PA}^\omega}

\newcommand{\CONT}{{\sf Cont}}

\newcommand{\AC}{{\sf AC}}

\newcommand{\DC}{{\sf DC}}

\newcommand{\DNS}{{\sf DNS}}

\newcommand{\OI}{{\sf OI}}

\newcommand{\R}{{\sf R}}

\newcommand{\wR}{{\sf wR}}
\newcommand{\wI}{{\sf wI}}
\newcommand{\IBR}{{\sf BarR}}

\newcommand{\BR}{{\sf BkR}}

\newcommand{\SBR}{{\sf SBR}}
\newcommand{\BBC}{{\sf BBC}}

\newcommand{\MBR}{{\sf MBR}}
\newcommand{\IPS}{{\sf IPS}}

\newcommand{\UR}{{\sf UpR}}
\newcommand{\OR}{{\sf Open}}

\newcommand{\zero}{0}

\newcommand{\mr}{{\; mr \;}}
\newcommand{\least}{\mu}

\newcommand{\at}{{\; \oplus \;}}

\newcommand{\update}[3]{{#1}_{#2}^{#3}}

\newcommand{\backI}{{\sf BkI}}

\newcommand{\dom}{{\rm dom}}

\newcommand{\PS}{{\sf PS}}

\newcommand{\oDC}[1]{{#1}\mbox{-}{\sf DC}_{\scriptsize \sf seq}}
\newcommand{\suc}{{succ}}

\newcommand{\dc}[2]{{#1}^{#2}}
\newcommand{\initSegss}[2]{\{#1\}({#2})}
\newcommand{\initSegsss}[2]{\langle #1 \rangle({#2})}

\newcommand{\negt}{N}

\begin{document}

\title{Parametrised bar recursion: A unifying framework for realizability interpretations of classical dependent choice}

\author{Thomas Powell\\ Institute of Computer Science, University of Innsbruck\footnote{\textbf{Email:} Thomas.Powell@uibk.ac.at \; \; \; \textbf{Tel:} +43 512 507 53293}}
\date{\today}
\maketitle

\begin{abstract}During the last twenty years or so a wide range of realizability interpretations of classical analysis have been developed. In many cases, these are achieved by extending the base interpreting system of primitive recursive functionals with some form of bar recursion, which realizes the negative translation of either countable or countable dependent choice. In this work we present the many variants of bar recursion used in this context as instantiations of a general, parametrised recursor, and give a uniform proof that under certain conditions this recursor realizes a corresponding family of parametrised dependent choice principles. From this proof, the soundness of most of the existing bar recursive realizability interpretations of choice, including those based on the Berardi-Bezem-Coquand functional, modified realizability and the more recent products of selection functions of Escard\'o and Oliva, follows as a simple corollary. We achieve not only a uniform framework in which familiar realizability interpretations of choice can be compared, but show that these represent just simple instances of a large family of potential interpretations of dependent choice principles.

\

\end{abstract}

\section{Introduction}
\label{sec-introduction}

One of the central problems of mathematical logic and computer science is to understand the constructive meaning of non-constructive proofs. From the early 20th century onwards a rich variety of constructive interpretations of classical logic have been developed, together with techniques for extracting computational content from proofs. The majority of these techniques initially deal with proofs in formal systems of classical predicate logic or weak subsystems of mathematics such as Peano arithmetic. In order to interpret stronger non-constructive proofs from subsystems of mathematical analysis which contain principles such as the axiom of dependent choice,
\begin{equation*}\DC \; \colon \; \forall n,x^\rho\exists y^\rho A_n(x,y)\to\exists f^{\NN\to \rho}\forall n A_n(f(n),f(n+1)),\end{equation*}
these techniques must typically be adapted and extended - a process which tends to be highly non-trivial.

In this paper we focus on just one method of giving a computational interpretation to proofs: namely classical modified realizability. A long established way of extending this technique to deal with choice principles is to introduce some kind of bar recursion to the usual interpreting system of primitive recursive functionals. Bar recursive interpretations of choice principles are among the most widely known and studied methods for giving a computational interpretation to classical analysis, and numerous instances of this method can be found in the literature. However, both the particular variant of bar recursion used and the form of choice that it realizes differs from case to case, and it is not well-understood how these variants compare as realizers. This is compounded by the fact that each variant was typically devised in a distinct setting, and both their existence and correctness often established using a slightly different method of proof, obscuring the inherent similarities between them all.

The main contribution of this paper is to construct a general bar recursive term which contains several free parameters. We show that whenever these parameters obey certain conditions, the term can be used to realize the negative translation of a correponding variant of countable dependent choice, and that moreover essentially all of the known variants of bar recursion used to extend modified realizability to classical analysis appear as simple instantiations of the parameters.

Our motivation for this is twofold. Firstly, we obtain a clear, unifying framework in which the differences between existing variants of bar recursion and their correctness proofs are essentially reduced to simple structures on the natural numbers, and thus their behaviour as realizers is more easily compared. Secondly, in doing this we greatly generalise the interpretation of classical analysis by giving not one but a whole class of realizers of choice principles, which can be freely chosen between to suit the problem at hand. Having a such a wide range of options in front of us means that in concrete instances of program extraction we are given the potential to tailor our realizer to the situation at hand and extract computational content which is more efficient and semantically meaningful.

\subsection{Realizability interpretations of classical analysis: A brief overview}
\label{sec-introduction-history}

This work is designed to be as self-contained as possible, and in particular no prior knowledge of bar recursion is assumed. However, because we are partly motivated by the desire to unify existing interpretations, we provide here a very short summary of some of the best known variants of bar recursion that we aim to bring together. We simply state without explanation their defining equations: the purpose here is not to treat in any detail the exact meaning of these objects, but to allow the reader to see certain key features of the recursion that will be dealt with in a more general setting later.

The idea of interpreting countable choice with bar recursion was first established not for realizability interpretations, but for G{\"o}del's slightly more intricate Dialectica interpretation, in a fundamental paper of Spector in 1962 \cite{Spector(1962.0)}. Here, bar recursion was given as the schema
\begin{equation}\label{eqn-intro-SBR}\SBR(Y,G,H,s^{X^\ast})=\begin{cases}G(s) & \mbox{if $Y(\ext{s})<|s|$}\\ H(s,\lambda x\; . \; \SBR(Y,G,H,s\ast x)) & \mbox{otherwise}.\end{cases}\end{equation}
The key details here are that $s\colon X^\ast$ is a finite sequence, $s\ast x$ denotes the extension of $s$ with the object $x$, $\ext{s}$ its canonical extension as an infinite sequence, and the output type of $Y$ is a natural number. Spector's bar recursion is therefore a form of `backward recursion' in which recursive calls are made on extensions $s\ast x$ of the argument $s$. The parameter $Y$ is responsible for terminating the recursion, based on the assumption that the underlying tree barred by sequences satisfying $Y(\ext{t})<|t|$ is well-founded - a fact that is true in all continuous models and also in non-continuous structures such as the strongly majorizable functionals \cite{Bezem(1985.0)}. What this means is that for any sequence $s\ast x_1,s\ast\pair{x_1,x_2},\ldots$ of recursive calls made in the computation of $\SBR(Y,G,H,s)$, there is always some $N$ such that $Y(\widehat{s\ast\pair{x_1,\ldots,x_N}})<|s|+N$, and thus $\SBR(Y,G,H,s)$ is a well-defined value.

An early adaptation of this idea to realizability was given by Berardi et al. \cite{BBC(1998.0)}, in which an interpretation for dependent choice broadly similar but somewhat simpler than Spector's variant of bar recursion was established. This realizer was later reformulated in the more standard framework of modified realizability in \cite{BergOli(2005.0)}, where it was given its now familiar name \emph{modified bar recursion} and was defined as
\begin{equation}\label{eqn-intro-MBR}\MBR(Y,H,s)=_\NN Y(s\at\lambda n^\NN.H(s,\lambda x\; . \; \MBR(Y,H,s\ast x))).\end{equation}
Here $s$ is again a finite sequence, $s\at\alpha$ denotes the infinite sequence $\alpha$ overwritten by $s$, and the outcome type of $Y$ is a natural number. This restriction on the type of $Y$ ensures that in continuous models $Y$ only looks at a finite amount of information from its input sequence. In other words, given a sequence of recursive calls $s\ast x_1,s\ast\pair{x_1,x_2},\ldots$ in the computation of $\MBR(Y,H,s)$, there is some point $N$ such that the value of $q(s\ast \pair{x_1,\ldots,x_N}\oplus\ldots)$ is determined based only on $s\ast \pair{x_1,\ldots,x_N}$, and so no further recursive calls are necessary. Thus, like Spector's variant of bar recursion (\ref{eqn-intro-SBR}), modified bar recursion is carried out over some underlying well-founded tree, although unlike Spector's bar recursion, this tree is not computable but it \emph{implicitly} well-founded by some kind of continuity argument.\footnote{Indeed, a fact closely related to this observation, proved in \cite{BergOli(2006.0)}, is that (\ref{eqn-intro-MBR}) is not S1-S9 computable in the total continuous functionals and is therefore strictly stronger than (\ref{eqn-intro-SBR}).}

In addition to what later became modified bar recursion, \cite{BBC(1998.0)} contains a striking realizer for (the negative translation of) countable choice, which took from Spector only the basic idea of backward recursion, replacing the sequential bar recursive calls $s\mapsto s\ast x$ of (\ref{eqn-intro-SBR}) and (\ref{eqn-intro-MBR}) with a symmetric updating $u\mapsto\update{u}{n}{x}$ of finite partial functions. This realizer was again put into a standard realizability framework in \cite{Berger(2002.0),Berger(2004.0)}, and will be referred to here as the Berardi-Bezem-Coquand, or BBC-functional. Its defining equation (following \cite{Berger(2002.0)}) is given by
\begin{equation}\label{eqn-intro-BBC}\BBC(Y,H,u)=_\NN Y(u\at\lambda n.H(n,\lambda x\; . \; \BBC(\update{u}{n}{x}))).\end{equation}
Here, $u$ is a finite partial function and $\update{u}{n}{x}$ is the domain-theoretic extension of $u$ with value $x$ at input $n$. As with modified bar recursion, $\BBC$ terminates by continuity of $Y$, although proving this is somewhat technical and as shown in \cite{Berger(2004.0)} is most elegantly done with Zorn's lemma in the form of open induction. It has been argued that (\ref{eqn-intro-BBC}) provides a computational interpretation of countable choice that is superior to standard bar recursion in that it is `demand driven', in the sense that to compute an approximation for the $n$th point in a choice sequence it is not automatically necessary to compute approximations for $1,\ldots,n-1$ first \cite{BBC(1998.0)}.

More recently, a family of new variants of bar recursion known as \emph{products of selection functions} were developed and explored by Escard{\'o} and Oliva, beginning in \cite{EscOli(2010.0)}. One of these, the so-called \emph{implicit product of selection functions}, was shown in \cite{EscOli(2012.0)} to realize not only the negative translation of countable dependent choice but more generally a dependent version of the so called $J$-shift arising from the Pierce translation. The product of selection functions is distinguished by the fact that it incorporates course-of-values recursion into its bar recursive calls. It can be formulated as
\begin{equation}\label{eqn-intro-IPS}\IPS(Y,H,s)=_\NN Y(s\at\lambda n. H(t_n,\lambda x\; . \; \IPS(Y,H,t_n\ast x)))\end{equation}
where $t_n$ is the sequence of length $n$ primitive recursively defined by 
\begin{equation*}(t_n)_i=s_i\mbox{ if $i<|s|$, else $H(t_i,\lambda x\; . \; \IPS(Y,H,t_i\ast x))$ for $i<n$}.\end{equation*}
In addition to solving the modified realizability interpretation of choice, this form of bar recursion has deep links to game theory as a functional that computes optimal strategies in a class of higher-type, continuously well-founded sequential games \cite{EscOli(2011.0)}. This provides a highly illuminating bridge between the computational content of the axiom of choice and the world of game theory, and is exploited in e.g. \cite{OliPow(2012.1)} to give a game-theoretic constructive interpretation of Ramsey's theorem.

Of course, this list is by no means exhaustive, and several further variants of (\ref{eqn-intro-MBR})-(\ref{eqn-intro-IPS}) have been devised for realizing choice principles. 
For example, a realizer of the refined $A$-translation of a seqential variant of dependent choice is given in \cite{Seisenberger(2008.0)}, and is used in \cite{Seisenberger(2003.0)} to extract a realizer for Higman's lemma. Similarly, forms of bar recursion closely related to (\ref{eqn-intro-MBR}) have been used for realizing dependent choice in a range of settings, including Parigot's $\lambda\mu$-calculus \cite{BloRib(2013.0)} and in the context of realizability toposes \cite{Streicher(2014.0)}.

Nevertheless, all of these recursors share two important features in common:

\begin{enumerate}

\item they take as input some partial object $t$, and make recursive calls over extensions of this object;

\item they terminate because the value of $Y(\alpha)$ only depends on a finite portion of the input $\alpha$.

\end{enumerate}

In fact, one could go as far as to say that (\ref{eqn-intro-MBR})-(\ref{eqn-intro-IPS}) along with their many other variants are essentially the same realizer, with the single exception that recursive calls are made using slightly different patterns. In this paper we make this idea precise, and show that the basic recipe used to form the realizers mentioned above can be used to construct, in a completely uniform way, an infinite class of bar recursive functionals, each one of which can be used to interpret a particular set of variants of the axiom of dependent choice.

Note that we restrict our attention to interpretations of analysis based on modified realizability, and in particular bar recursors of the form (\ref{eqn-intro-SBR}) which arise from the Dialectica interpretation do not fit into our framework. Nevertheless, we believe that the basic ideas behind this work could be readily lifted to the Dialectica interpretation (see for example the author's recent article with P. Oliva \cite{OliPow(2014.0)} on a symmetric form of (\ref{eqn-intro-SBR}) which updates partial funtions similarly to the BBC-functional), and may also be helpful in linking proof interpretations to more direct computational interpretations of analysis, such as the learning-based realizabilities of e.g. \cite{Aschieri(2011.1),Avigad(2002.0)}.

\subsection{Outline of the paper}
\label{sec-introduction-outline}

After setting up our basic formal systems in Section \ref{sec-formal}, we define in Section \ref{sec-open} a general principle of backward induction which will be used throughout the paper in order to prove the correctness of our realizers. Backward induction replaces principles such as bar induction or dependent choice which are more typically used on the meta-level to prove correctness of realizers, and is preferred here due to its far greater flexibility. We also define an analogous schema of backward recursion which will be used to construct our parametrised realizers. This section is strongly influenced by Berger \cite{Berger(2004.0)}, and our formulation of backward induction is similar to but slightly more general than his update induction.

Sections \ref{sec-countable} and \ref{sec-dependent} form the core of the paper, in which we define our parametrised realizer and use it to realize dependent choice. In Section \ref{sec-countable} we restrict ourselves to the simpler double negation shift: While this section's main result, Theorem \ref{thm-main-AC}, is eventually subsumed by the later Theorem \ref{thm-main-DC}, it is instructive to first focus on the double negation shift so that the main ideas of Section \ref{sec-dependent} can be appreciated. In the sequel we introduce a family of dependent choice principles parametrised by a well-founded ordering $\lhd$ on the natural numbers, and show that the negative translation of each of these priniciples can be realized by our parametrised bar recursor under certain conditions. Theorem \ref{thm-main-DC} gives essentially all the results listed in Section \ref{sec-introduction-history} as a Corollary, and vastly generalises all of them.

We conclude on a more informal level, and in Section \ref{sec-semantics} discuss a potential semantic interpretation of our parametrised realizer, inspired by the game-theoretic semantic reading of the Berardi-Bezem-Coquand functional in \cite{BBC(1998.0)}.

\section{Preliminaries}
\label{sec-formal}

Throughout this paper we work in variant of extensional Heyting/Peano arithmetic in all finite types, which will be essentially the $\EHAomega$/$\EPAomega$ as defined in e.g. \cite{Kohlenbach(2008.0),Troelstra(1973.0)}, but based on a slightly richer type system.

\subsection{The finite types}
\label{sec-formal-types}

For us, the finite types consist of base types $\N$ and $\B$ for natural numbers and booleans, and are build from the formation of function types $\rho\to\tau$, product types $\rho\times\tau$, and finite sequence types $\rho^\ast$. We sometimes use the abbreviation $\tau^\rho$ for $\rho\to\tau$.

A \emph{discrete} type is any type $\tau$ which can be encoded in $\N$, and so in all standard models $\B$, $\N$, $\B\times \N$, $\N^\ast$ etc. are discrete types but $\N\to \N$ is not. 

\subsection{The theory $\EHAomega$}
\label{sec-formal-theory}

Terms of $\EHAomega$ are typed lambda terms formed by application and abstraction, and include variables for each type, the usual constructors and deconstructors for product and sequence types, the arithmetic constants $0\colon \N$ and $s\colon \N\to \N$, and finally recursors $\R_\rho\colon\rho\to (\N\to\rho\to\rho)\to\N\to\rho$ of each type.

Equations in $\EHAomega$ are formed using basic symbols $=_\B$ and $=_\N$ for equality of type $\B$ and $\N$, while formulas of $\EHAomega$ are built from $=_\N$ and $=_\B$ together with the usual logical connectives and quantifiers for all types. Equality at arbitrary types is defined inductively, so for example $f=_{\rho\to\tau} g$ abbreviates $\forall x^\rho(fx=_\tau gx)$. 

The axioms and rules of $\EHAomega$ are the standard axioms and rules of classical logic in all finite types, along with those of the typed lambda calculus, defining equations for all the constants, induction for arbitrary formulas, and finally full extensionality:
\begin{equation*}\forall f^{\rho\to\tau}\forall x^\rho,y^\rho(x=_\rho y\to fx=_\tau fy).\end{equation*}
In what follows we often consider extensions of $\EHAomega$ with new recursively defined functionals $\F$, in which case by $\EHAomega+\F$ we mean the extension of $\EHAomega$ with new  constants $\F$ and their associated defining axioms.

We use, throughout, the following notation and abbreviations.
\begin{itemize}

\item $0_\rho$ is the inductively defined $0$ term of type $\rho$.

\item $\pi_i$ and $\pair{\;,\;}$ denote the projection and pairing operations. For a sequence $\alpha\colon (\rho_0\times\rho_1)^\N$ we sometimes write $\alpha_i\colon\rho^\N$ for $\lambda n.\pi_i(\alpha(n))$.

\item For a finite sequence $s\colon\rho^\ast$, $|s|$ denotes the length of $s$, while $s\ast t:=\pair{s_0,\ldots,s_{m-1},t_0,\ldots,t_{n-1}}$ denotes the concatenation $s$ and $t$. We also use $s\ast x$ to denote $s\ast\pair{x}$, and $s\ast\alpha$ for the concatenation of $s$ with the infinite sequence $\alpha$.

\item $s\at\alpha\colon\rho^\N$ denotes the overwriting of $\alpha$ with the finite sequence $s$ i.e. $$(s\at\alpha)(n):=\begin{cases}s_n & \mbox{if $n<|s|$}\\ \alpha(n) & \mbox{if $n\geq |s|$}\end{cases}.$$

\item $\initSeg{\alpha}{n}:=\pair{\alpha(0),\ldots,\alpha(n-1)}$ denotes the finite initial segment of length $n$ of $\alpha\colon\rho^\N$, while $\ext{s}$ denotes the canonical extension $s\oplus \zero_{\rho^\N}$ of the finite sequence $s\colon\rho^\ast$.

\item For a decidable predicate $P(x)$, the term $`y^\rho\mbox{ if $P(x)$}'$ of type $\rho$ is shorthand for $$\begin{cases}y & \mbox{if $P(x)$}\\ 0_\rho & \mbox{otherwise}\end{cases}.$$ 

\item For a decidable predicate $P(n)$ on $\N$, the term $\least i\leq n.P(i)\colon\N$ is the least $i\leq n$ satisfying $P(i)$, and just $n$ if no such $i$ exists. 

\end{itemize}
In addition to all this, throughout the article we work with a type of partial sequences with potentially infinite domain. These can be encoded, for example, by total sequences of type $\N\to\Bool\times\rho$, where defined values are represented as $(1,x)$, and undefined values as $(0,0_\rho)$. Accordingly, for an object $u$ of this type we say that $n$ is in the domain of $u$, or $n\in\dom(u)$, if $\pi_0 u(n)=1$, and $n\notin\dom(u)$ otherwise. Membership of $\dom(u)$ is a decidable property. 

We imagine $\Bool\times\rho$ as simulating a type $\ps{\rho}\equiv\rho+\one$, and we write $u(n)=_{\ps{\rho}}x$ instead of $u(n)=_{\Bool\times\rho}(1,x)$, and $u(n)=_{\ps{\rho}}\bot$ whenever $n$ is not in the domain of $u$. Similarly, $u=_{\ps{\rho}} v$ if for all $n$ we have either $\pi_0u(n)=\pi_1u(n)=0$ or $\pi_0u(n),\pi_0v(n)=1$ and $\pi_1u(n)=\pi_1v(n)$.
\begin{itemize}

\item We use the symbol $\emptyset$ to denote the partial function with empty domain.

\item We extend the overwrite operation $\at$ given for finite sequences above to partial functions by defining $u\at v\colon\ps{\rho}^\N$ by
\begin{equation*}\label{eqn-overwrite}(u\at v)(n)=\begin{cases}u(n) & \mbox{if $n\in\dom(u)$}\\ v(n) & \mbox{if $n\notin\dom(u)$}.\end{cases}\end{equation*}
We define $u\at\alpha\colon\rho^\N$ for total sequence $\alpha^\N$ analogously. It will always be clear from the context which types the operator $\at$ takes as input.

\item We isolate as a special case the addition of a single value to $u$: For $n^\N$ and $x^\rho$ we define $\update{u}{n}{x}=u\at (n,x)$ where $(n,x)$ is the partial function taking defined value $x$ at point $n$, and undefined elsewhere. When $n\notin\dom(u)$ we say that $\update{u}{n}{x}$ is an \emph{update} of $u$.

\item Finally, we write $u\sqsupseteq v$ whenever $\forall i\in\dom(v)(u(i)=v(i))$. Note that we have $u\at v\sqsupseteq u$ for any $v$, and $u\at v\sqsupset u$ whenever there is some $n\notin\dom(u)$ such that $n\in\dom(v)$. We also have $\update{u}{n}{x}\sqsupset u$ when $n\notin\dom(u)$.

\end{itemize}
For us, the system $\EHAomega$ together with its classical variant $\EPAomega$ acts as a standard lambda calculus equipped with a robust meta-theory for reasoning about terms. However, the exact details of our formal system are not particularly important as everything which follows can be easily lifted to alternative settings. For example, a slightly different approach would be to work a weaker, quantifier-free term calculus $\sf T$ and to do all the reasoning in an unspecified meta-theory, as in \cite{BBC(1998.0)}. Alternatively we could work in a theory of partial continuous functionals as in \cite{Berger(2002.0)}, taking our base type to represent the flat domain $\N_{\bot}$.

\subsection{Models of $\EHAomega$}
\label{sec-formal-models}

In order to prove both the existence and correctness of bar recursive realizers, it is typically necessary to work in a constructive interpretation $T:\equiv (T_\rho)$ of $\EHAomega$ which satisfies some form of the following two properties:
\begin{enumerate}

\item $(T_\rho)$ contains \emph{arbitrary} choice sequences, in other words $T_{\N\to\rho}$ contains all sequences $\NN\to T_\rho$ and so in particular $(T_\rho)$ validates dependent choice;

\item Whenever $\tau$ is a discrete type, functionals of type $F\colon\rho^\N\to \tau$ satisfy the following continuity principle:
\begin{equation*}\CONT \; \colon \; \forall\alpha^{\rho^\N}\exists n\forall\beta(\initSeg{\alpha}{n}=_{\rho^\ast}\initSeg{\beta}{n}\to F\alpha=_\tau F\beta). \end{equation*}

\end{enumerate}
Both of these principles are satisfied automatically by the Kleene-Kreisel continuous functionals $\modcont$ \cite{Kleene(1959.0),Kreisel(1959.0)}, whereas for term models such as the theory $\mathcal{P}$ of \cite{BBC(1998.0)}, (1) is obtained by adding infinite choice sequences explicitly. Here we do not choose any particular intepretation of $\EHAomega$, rather we simply add principles such as $\CONT$ and dependent choice to our meta-theory whenever they are required.

\section{Backward induction and recursion}
\label{sec-open}

We now develop some of the crucial background theory that will be required in order to prove our main results. In particular, we formulate a general principle of backward induction and define an associated backward recursor, both of which will be used to construct and verify the correctness of our parametrised realizer. In simple terms, backward induction is induction over domain-theoretic extensions of partial sequences, and should be seen as a generalisation of bar induction. Analogously, backward recursion is a generalisation of the implicit forms of bar recursion used to give realizability interpretations to countable choice. This section is largely inspired by \cite{Berger(2004.0)} in that we formulate backward induction as an instance of the still more general principle of open induction. 

\subsection{Open induction}
\label{sec-open-open}

Open induction, first considered by Raoult in \cite{Raoult(1988.0)}, is an extension of well-founded (or Noetherian) induction to chain-complete partial orders. Recall that a partial order $(X,\leq)$ is chain-complete if every non-empty chain $\gamma$ in $X$ has a least upper bound $\bigvee \gamma$. A predicate $B$ on $X$ is \emph{open} if it satisfies the property
\begin{equation*}B(\bigvee \gamma)\to\exists x\in\gamma B(x)\end{equation*}
for every non-empty chain $\gamma$ in $X$, and the principle of open induction over $X$ is given by
\begin{equation*}\OI_{(X,\leq)} \; \colon \; \forall x(\forall y>x B(y)\to B(x))\to\forall x B(x) \end{equation*}
where $B$ ranges over open predicates. Note that open induction implies well-founded induction since whenever $>$ is well-founded $(X,\leq)$ is trivially chain-complete and all predicates are automatically open. However, in general $>$ need not be well-founded, in which case openness becomes a non-trivial property.

\begin{theorem}[Raoult \cite{Raoult(1988.0)}]\label{thm-Zorn-OI}Any chain-complete partial order satisfies open induction.\end{theorem}

\begin{proof}This is a direct consequence of Zorn's lemma. Suppose that the open predicate $B$ satisfies the premise of open induction, which is classically equivalent to
\begin{equation}\label{eqn-Zorn-OI}\forall x(\neg B(x)\to \exists y>x\neg B(y)),\end{equation}
and suppose for contradiction that the set
\begin{equation*}S=\{x\in X \; | \; \neg B(x)\}\end{equation*}
is non-empty. We show that every chain in $S$ has an upper bound in $S$. For the empty chain this is trivial since $S$ contains at least one element. On the other hand, if $\gamma$ is non-empty, then it has an upper bound $\bigvee\gamma$ in $X$ by completeness, and moreover $\bigvee\gamma\in S$ since $\forall x\in\gamma\neg B(x)\to\neg B(\bigvee\gamma)$ by openness of $B$. Therefore by Zorn's lemma $S$ contains a maximal element, which contradicts (\ref{eqn-Zorn-OI}). Thus $S$ must be empty.\end{proof}

\subsection{Backward induction}
\label{sec-open-backward}

We are now ready to define backward induction, which we simply take to be open induction over the partial order $(\N\to\ps{\rho},\sqsubseteq)$, where $\sqsubseteq$ is the extension relation on partial sequences defined in Section \ref{sec-formal-theory}. In other words, backward induction is the schema
\begin{equation*}\backI \ \colon \ \forall u^{\N\to\ps{\rho}}(\forall v\sqsupset u B(v)\to B(u))\to\forall u B(u) \end{equation*}
where $B$ ranges over open predicates. This formulation of backward induction makes sense and is valid in any interpretation $T$ of $\EPAomega$ that admits arbitrary sequences, since for any chain $\gamma$ in $T_{\N\to\ps{\rho}}$,
\begin{equation*}(\bigvee\gamma)(n):=\begin{cases}u(n) & \mbox{if $n\in\dom(u)$ for some $u\in\gamma$}\\ \bot & \mbox{otherwise}.\end{cases}\end{equation*}
is a perfectly well-defined object of $\NN\to T_{\ps{\rho}}\equiv T_{\N\to\ps{\rho}}$ and is the least upper bound of $\gamma$ with respect to $\sqsubseteq$. It will be convenient to isolate the following syntactic notion of an open formula, which will be sufficient for everything that follows.

\begin{proposition}\label{prop-openp}Suppose that the formula $B(u)$ on partial sequences is of the form 
\begin{equation*}B(u):\equiv \forall n[n\in\dom(u)\to A(n,u)]\to \exists nP(\initSeg{u}{n})\end{equation*}
where $P$ is an arbitrary predicate on $\ps{\rho}^\ast$, and $A$ is monotone in the following sense
\begin{equation}\label{eqn-openp-mon}n\in\dom(u)\wedge u\sqsubseteq v\wedge A(n,u)\to A(n,v).\end{equation}
Then $B(u)$ is open with respect to $\sqsubseteq$.\end{proposition}

\begin{proof}Given some chain $\gamma$ let $v=\bigvee\gamma$ and assume that $B(v)$ holds. We must show that $B(u)$ holds for some $u\in\gamma$. If $B(v)$ holds, then by classical logic we either have $P(\initSeg{v}{n})$ for some $n$, in which case $B(u)$ holds for e.g. the least $u\in\gamma$ satisfying $\initSeg{u}{n}=\initSeg{v}{n}$, or we have
\begin{equation*}n\in\dom(v)\wedge\neg A(n,v)\end{equation*}
for some $n$. In this latter case pick the least $u\in\gamma$ satisfying $n\in \dom(u)$. Then since $u\sqsubseteq v$, we have $A(n,u)\to A(n,v)$ by the monotonicity condition, which is a contradiction. Therefore $\neg A(n,u)$ and $B(u)$ holds. \end{proof}

\begin{remark}Note that Berger \cite{Berger(2004.0)} generally studies formulas which are open with respect to the \emph{lexicographic ordering} on infinite sequences, which are not necessarily open with respect to $\sqsubseteq$ (which is strictly contained in a lexicographic ordering). However, any formula equivalent to one of the form $B(u):\equiv \exists n P(\initSeg{u}{n})$ is open in both senses.  \end{remark}

\subsection{Some additional remarks on backward induction}
\label{sec-open-backward}

We take the opportunity to explore backward induction in more detail and relate it to other well-known principles in logic, including the minimal bad sequence argument and bar induction. This section is not strictly necessary for the remainder of the paper, so if the reader prefers they can proceed directly to the definition of backward recursion given in Setion \ref{sec-open-recursion}.

We first point out that, analogously to \cite{Berger(2004.0)}, backward induction does not require the full strength of Zorn's lemma, and is provable from just dependent choice, using a version of the minimal-bad-sequence argument due to Nash-Williams \cite{NashWilliams(1963.0)}. 

\begin{proposition}The principle of backward induction is provable in $\EPAomega+\DC$.\end{proposition}

\begin{proof}Take some open formula $B(u)$, which for simplicity we assume is open in the sense of Proposition \ref{prop-openp}. Suppose for contradiction that we have $\forall u(\neg B(u)\to\exists v\sqsupset u\neg B(v))$ but there exists some partial sequence $u_0$ such that $\neg B(u_0)$. Using dependent choice construct the sequence $(u_n)$ as follows: Supposing that we have already constructed $\pair{u_0,\ldots,u_n}$ for $n\geq 0$, define
\begin{enumerate}[(i)]

\item\label{item-MBSi} $u_{n+1}:=w$ if $n\notin\dom(u_n)$ and $w$ is such that the following four properties are satisfied: $\neg B(w)$; $w\sqsupset u_n$; $\initSeg{w}{n}=\initSeg{u_n}{n}$ and $n\in\dom(w)$,

\item\label{item-MBSii} $u_{n+1}:=u_n$ if either $n\in\dom(u_n)$ or $n\notin\dom(u_n)$ and no $w$ in the sense of (\ref{item-MBSi}) exists.

\end{enumerate}
First, it is clear by a simple induction that for all $n$ we have 
\begin{enumerate}[(a)]

\item\label{item-MBSpi} $\neg B(u_n)$,

\item\label{item-MBSpii}$\initSeg{u_n}{n}=\initSeg{u_{n+1}}{n}$, and

\item\label{item-MBSpiii}$u_n\sqsubseteq u_{n+1}$. 
\suspend{enumerate}
Define $\tilde u:=\lambda n.u_{n+1}(n)$. Then it follows that for all $n$ we have
\resume{enumerate}[{[(a)]}]

\item\label{item-MBSpiv} $\initSeg{\tilde u}{n}=\initSeg{u_n}{n}$ and 

\item\label{item-MBSpv} $u_n\sqsubseteq \tilde u$. 

\end{enumerate}
The first of these is done by a simple induction using (\ref{item-MBSpii}). For the latter, take $i\in\dom(u_n)$. Then either $i< n$ in which it is clear by (\ref{item-MBSpv}) that $\tilde u(i)=u_n(i)$, or $i\geq n$ and we obtain $\tilde u(i)=u_{i+1}(i)=u_n(i)$ by $u_n\sqsubseteq\ldots\sqsubseteq u_{i+1}$.

Now we prove that $\neg B(\tilde u)$, which is classically equivalent to
\begin{equation*}\forall i([i\in\dom(u)\to A(i,u)]\wedge\neg P(\initSeg{u}{i})).\end{equation*}
Taking some arbitrary $n$, and setting $u=u_{n+1}$ and $i=n$ we get, by $\neg B(u_{n+1})$ (true by (\ref{item-MBSpi})),
\begin{equation*}(\ast) \ \ [n\in\dom(u_{n+1})\to A(n,u_{n+1})]\wedge \neg P(\initSeg{u_{n+1}}{n}).\end{equation*}
But $\neg P(\initSeg{u_{n+1}}{n})\to\neg P(\initSeg{\tilde u}{n})$ by (\ref{item-MBSpii}) and (\ref{item-MBSpiv}), and furthermore $n\in\dom(\tilde u)$ is equivalent to $n\in\dom(u_{n+1})$ and hence by $(\ast)$ implies $A(n,u_{n+1})$, so using monotonicity of $A$ and the fact that $u_{n+1}\sqsubseteq \tilde u$ we have $A(n,\tilde u)$ and have therefore established
\begin{equation*}n\in\dom(\tilde u)\to A(n,\tilde u).\end{equation*}
Taken together and bearing in mind that $n$ is arbitrary, this implies $\neg B(\tilde u)$. But now we know by the backward induction hypothesis that there exists some $v\sqsupset\tilde u$ such that $\neg B(v)$ holds. 

We can show that this contradicts the construction of $(u_n)$, and therefore there cannot exists any initial sequence $u_0$ satisfying $\neg B(u_0)$, and we're done. Let $m$ be the least point such that $m\in\dom(v)$ but $m\notin\dom(\tilde u)$. Then firstly by (\ref{item-MBSpv}) we have $u_m\sqsubseteq\tilde u\sqsubset v$ and thus $u_m\sqsubset v$, and secondly $\initSeg{v}{m}=\initSeg{\tilde u}{m}=\initSeg{u_m}{m}$, the first equality by minimality of $m$ and the second by (\ref{item-MBSpiv}), and therefore $v$ satisfies the required properties of $w$ in (\ref{item-MBSi}) at point $m$. In addition we know that $m\notin\dom(u_{m})$ else we'd have $m\in\dom(\tilde u)$ by (\ref{item-MBSpv}), and therefore $u_{m+1}$ must be constructed using (\ref{item-MBSi}) and thus $m\in\dom(u_{m+1})\subseteq \dom(\tilde u)$, contradicting the assumption that $m\notin\dom(\tilde u)$.  \end{proof}

As we will see in Sections \ref{sec-countable} and \ref{sec-dependent}, one of the key ideas in this paper is construct forms of recursion based on restricted, or relativised variants of backward recursion which take as input partial functions that are downward closed with respect to some relation on $\N$. These will be closely related to the following form of relativised backward induction.

\begin{proposition}Let $\lhd$ be some decidable relation on $\N$, and define the predicate $u\in D_\lhd$ by
\begin{equation*}u\in D_\lhd:\equiv\forall n\in\dom(u)[\forall i\lhd n(i\in\dom(u))].\end{equation*}
Equivalently, we say that $\dom(u)$ is $\lhd$-closed. Then for any relation $\sqsubset'$ on $\ps{\rho}^\N$ such that $u\sqsubset' v\to u\sqsubset v$, the following principle of relativised backward induction is provable from $\backI$:
\begin{equation*}\forall u\in D_\lhd(\forall v\sqsupset' u [v\in D_\lhd\to B(v)]\to B(u))\to\forall u\in D_\lhd \; B(u). \end{equation*}\end{proposition}

\begin{proof}First note that $u\in D_\lhd$ is of the form $\forall n\in\dom(u) D_0(n,u)$ with $D_0(n,u)$ monotone in the sense of (\ref{eqn-openp-mon}), therefore the predicate $B'(u):\equiv u\in D_\lhd\to B(u)$ is open for any open $B(u)$. Thus we obtain
\begin{equation*}\forall u(\forall v\sqsupset' u B'(v)\to B'(u))\to \forall u(\forall v\sqsupset u B'(v)\to B'(u))\to\forall u B'(u) \end{equation*}
the first implication following from the inclusion $\sqsubset'\subseteq\sqsubset$ and the second from normal backward induction applied to $B'(u)$. Rearranging this gives us relativised backward induction.\end{proof}

We can now instantiate $\lhd$ and $\sqsubset'$ to obtain certain well-known instances of backward induction.

\begin{example}[Update induction]If we define $u\sqsubset' v$ iff $v$ is an update of $u$, and let $\lhd$ just be the empty relation, then relativised backward induction just becomes update induction in the sense of \cite{Berger(2004.0)}. \end{example}

\begin{example}[Bar induction]\label{ex-bar-ind}Now consider the case $\lhd=<$. By classical logic, if $u\in D_<$ then either $u=\alpha$ for some total sequence $\alpha$ or $u=\ext{s}$ where $s$ is some finite sequence and $\ext{s}$ its embedding as a partial function. Therefore relativised backward induction is equivalent in $\EPAomega$ to
\begin{equation}\label{eqn-bar-ind0}\forall \alpha B(\alpha)\wedge \forall s(\forall v\sqsupset' \ext{s}[v\in D_<\to B(v)]\to B(\ext{s}))\to\forall s B(\ext{s}) \end{equation}
In addition, if we define $u\sqsubset' v$ iff $v$ is an update of $u$, then (\ref{eqn-bar-ind0}) becomes equivalent to
\begin{equation}\label{eqn-bar-ind1}\forall \alpha B(\alpha)\wedge \forall s(\forall x B(\ext{s\ast x})\to B(\ext{s}))\to\forall s B(\ext{s}) \end{equation}
which is just a variant of bar induction. It is not too difficult to show that this is equivalent over $\EPAomega$ to the more standard formulations of bar induction found in e.g. \cite{Troelstra(1973.0)}. 
%
%
\end{example}

\subsection{Backward recursion}
\label{sec-open-recursion}

The purpose of introducing backward induction was to give us a way to reason about backward recursion, which we define and discuss in this section. In the same way that backward induction is a special case of open induction, backward recursion is closely related (and in fact definable from) open recursion as defined in \cite{Berger(2004.0)}. 

To begin with, in order to motivate what follows let us consider as a comparison the entirely standard concept of well-founded recursion over some decidable well-founded relation $\prec$ on $\rho$. Assuming we are working in a structure such as PCF or the Scott continuous functionals, we can define a well-founded recusor $\wR_{\prec}$ as the fixpoint of the following recursive equation
\begin{equation*}\wR_{\prec}^\psi(x)=_\sigma \psi_x(\lambda y\; . \; \wR_{\prec}^g(y)\mbox{ if $y\prec x$}),\end{equation*}
and prove that the recursor defines a \emph{total} functional for any outcome type $\sigma$ using well-founded induction over $\prec$:
\begin{equation*}\wI \; \colon \; \forall x(\forall y\prec x A(y)\to A(x))\to\forall x A(x).\end{equation*}
We want to define a backward recursor in a similar way - although we have two problems: firstly the relation $u\sqsubset v$ is not decidable, and secondly backward induction is only valid for open formulas. We avoid these issues by defining backward recursion to be the fixpoint of the following recursive equation
\begin{equation*}\BR_{\rho,\tau}^\psi(u)=_\tau\psi_u(\lambda n,v\;. \; \BR^\psi(u\at v)\mbox{ if $n\in\dom(v)\backslash\dom(u)$}), \end{equation*}
where $\tau$ is restricted to being a discrete type, while $n\in\dom(v)\backslash\dom(u)$ denotes the decidable predicate $n\in\dom(v)\wedge n\notin\dom(u)$. Observe that any $w\sqsupset u$ is of the form $w=u\at v$ for some $v$ which is defined at at least one point $n\notin\dom(u)$, and so $\BR$ makes recursive calls over all $w\sqsupset u$, although crucially it must always have access to a point $n\in\dom(w)$ such that $n\notin\dom(u)$. The necessity of the restriction on $\tau$ is to ensure that totality of $\BR^\psi(u)$ is an open property on total input $u$ (here we are referring to total elements of the \emph{model} of partial continuous functionals as opposed to the type $\ps{\rho}$). Indeed for discrete $\tau$ and total $u$ we have
\begin{equation*}\BR^\psi(u)\mbox{ is total}\leftrightarrow\exists n\forall w(\BR^\psi(\initSeg{u}{n}\at w)\mbox{ is total})\end{equation*}
assuming sequential continuity $\CONT$ for functionals with total output. Therefore totality of $\BR$ is provable using backward induction. Note that alternatively, a direct proof via Zorn's lemma that $\BR$ exists as a total element of the Scott partial continuous functionals can be carried out using the same manner as the proof of totality of the Berardi-Bezem-Coquand functional in \cite{Berger(2002.0)}. 

Alternatively, one can justify the existence of backward recursion in continuous models by showing that backward recursion is definable from the slightly more general schema of open recursion on the lexicographic ordering considered in \cite{Berger(2004.0)}. Open recursion is defined to be the fixpoint of the following recursive equation:
\begin{equation*}\OR_{\rho,\tau}^{\psi}(u)=_\tau \psi_u(\lambda n,v\; . \; \OR^\psi(\initSeg{u}{n}\at v)\mbox{ if $n\in\dom(v)\backslash\dom(u)$})\end{equation*}
where once again $\tau$ is discrete, and lexicographically open recursive functionals of the above form are shown to be total in \cite[Proposition 5.1]{Berger(2004.0)}.

\begin{proposition}\label{def-open-back}$\BR$ is instance-wise primitive recursively definable from $\OR$.\end{proposition}

\begin{proof}Primitive recursively define 
\begin{equation*}m_{n,u,v}:=\mbox{least $i\leq n$ s.t. $i\in\dom(v)\backslash\dom(u)$, else $n$},\end{equation*}
and set $\BR_{\rho,\tau}^\psi(u)=\OR_{\rho,\tau}^{\tilde\psi}(u)$ where
\begin{equation*}\tilde\psi_u(f^{\N\times\ps{\rho}^\N\to\tau}):=\psi_u(\lambda n,v\; . \; f(m_{n,u,v},u\at v)\mbox{ if $n\in\dom(v)\backslash\dom(u)$}).\end{equation*}
Then expanding definitions we have
\begin{equation*}\begin{aligned}\BR^\psi(u)&=\tilde\psi_u(\lambda n,v\; . \; \BR^{\tilde\psi}(\initSeg{u}{n}\at v)\mbox{ if $n\in\dom(v)\backslash\dom(u)$})\\
&\stackrel{(a)}{=}\psi_u(\lambda n,v\; . \; \BR(\initSeg{u}{m_{n,u,v}}\at (u\at v))\mbox{ if $n\in\dom(v)\backslash\dom(u)$})\\
&\stackrel{(b)}{=}\psi_u(\lambda n,v\; . \; \BR(u\at v)\mbox{ if $n\in\dom(v)\backslash\dom(u)$})\end{aligned}\end{equation*} 
where for $(a)$ we use $n\in\dom(v)\backslash\dom(u)\to m_{n,u,v}\in\dom(u\at v)\backslash \dom(u)$, and $(b)$ follows by minimality of $m_{n,u,v}$. \end{proof}

\begin{remark}[Update recursion]\label{ex-update-rec}It is a fairly easy observation that update recursion as defined in \cite{Berger(2004.0)} -
\begin{equation*}\UR^H(u)=H_u(\lambda n,x^\rho\; . \; \UR^H(\update{u}{n}{x})\mbox{ if $n\notin\dom(u)$}),\end{equation*}
is a simple instance of backward recursion obtained by setting
\begin{equation*}\psi_u(f^{\N\times\ps{\rho}^\N\to\tau}):=H_u(\lambda n,x\; . \; f(n,\update{u}{n}{x})\mbox{ if $n\notin\dom(u)$}).\end{equation*}\end{remark}

\begin{remark}[Bar recursion]\label{ex-bar-rec}For those readers interested in the computability theory of bar recursion, it might be instructive to pause for a moment to consider a natural instance of bar recursion that arises from backward recursion. Let us define $\IBR(H,s^{\rho^\ast}):=\BR^\psi(\ext{s})$ where
\begin{equation*}\psi_u(f^{\N\times\ps{\rho}^\N\to\tau}):=_{\tau} H(u,\lambda t\; . \; f(|t|-1,\ext{t})). \end{equation*}
Then it is not too hard to show that $\IBR$ satisfies
\begin{equation*}\IBR(H,s)=H(\ext{s},\lambda t\;.\;\IBR(H,s\at t)\mbox{ if $|t|>|s|$}),\end{equation*}
and this can be viewed as a `implicitly well-founded' variant of Spector's bar recursion (\ref{eqn-intro-SBR}). The reason we highlight this is that while several such implicit variants of Spector's so-called `special' instance of bar recursion\footnote{See \cite{OliPow(2012.2)} for the distinction between the special and generals forms of Spector's bar recursor.} have been studied, including both modified bar recursion and the implicit product of selection functions, constructing a direct analogue to the general form is more complicated (for example, an implicit form of the so-called product of quantifiers is known not to exist \cite{EscOli(2011.0)}). 

The subtle reason for this is that such variants of bar recursion must not be allowed to access the length of the input sequence $s$. For example, no object $\Phi$ can satisfy the slightly altered equation
\begin{equation*}\Phi(H,s)=H(s,\lambda t\;.\;\Phi(H,s\at t)\mbox{ if $|t|>|s|$})\end{equation*}
even for discrete output type, since we could just take $\tau=\N$ and define $H(s,g):=1+f(s\ast 0)$, and then $\Phi(\pair{})=n+1+\Phi(H,\initSeg{0}{n+1})>n$ for all $n$, which cannot hold in any model of arithmetic. Indeed, trying to define this from $\BR$ with continuous $\psi_u$ is impossible, since we'd require a non-continuous unbounded search (and thus totality of the underlying instance of $\BR$ would no longer be an open predicate). Thus we overcome the difficulty with implicit variants of bar recursion by removing access to the length of the input. Note that this problem is not a feature of modified bar recursion and implicit products of selection functions (or indeed any of the realizers we define in the following sections), since these are defined `pointwise', and make recursive calls only when we are accessing points already greater than the length of the input sequence.\end{remark}

\section{A computational interpretation of the double negation shift}
\label{sec-countable}

Now that we have completed the mathematical groundwork we come to the core of the paper. In this section we give a new, general realizability interpretation to the double negation shift. Ultimately, this will form a special case of the interpretation of full dependent choice given in the next section. However, by focusing first on the double negation shift we have an opportunity to present our main ideas in a slightly simplified setting, then the extension to full dependent choice will mostly be a matter of taking care of some additional technical details. 

\subsection{Modified realizability intepretation of extensions of $\PAomega$}
\label{sec-countable-mr}

We begin by very briefly recalling how Kreisel's modified realizability can be used in conjunction with the so-called Friedman trick to extract programs from classical proofs of $\Pi^0_2$-formulas. This is all completely standard, so we omit most of the details. For every formula in the language of $\HAomega$ the realizability relation $x\mr A$ is inductively defined by
\begin{equation*}\begin{aligned}()\mr A &\equiv A\mbox{ if $A$ is an atomic formula},\\
x,y\mr (A\wedge B)&\equiv x\mr A\wedge y\mr B,\\
n^\N,x,y\mr (A\vee B)&\equiv (n=0\to x\mr A)\wedge (n\neq 0\to y\mr B),\\
f\mr (A\to B)&\equiv \forall x(x\mr A\to fx\mr B),\\
x\mr \forall zA(z)&\equiv\forall z(xz\mr A(z)),\\
x,y\mr\exists zA(x)&\equiv y\mr A(x).\end{aligned}\end{equation*}
It is well-known that whenever $\HAomega\vdash A$ then $\HAomega\vdash t\mr A$ where $t$ is some primitive recursive term extracted from the proof of $A$. The interpretation of classical logic, on the other hand, is more subtle. A simple combination of the negative translation with modified realizability fails to work since the atomic formula $\bot$ is realized by $()$ and therefore all negated formulas are trivially interpreted. In particular, this method gives us no way of extracting realizers from $\Pi^0_2$-formulas $\forall x^\N\exists y^\N A(x,y)$. 

One well established way of overcoming this problem is to slightly alter the definition of modified realizability by regarding $x^\N\mr\bot$ as an uninterpreted formula. Then, as discussed in e.g. \cite{BergOli(2005.0),BergSch(1995.0)}, from a classical derivation $\PAomega\vdash\forall y^\N\exists x^\N A(y,x)$ one can extract a term $t$ such that $\HAomega\vdash\forall y A(y,ty)$, utilising the aforementioned Friedman trick of replacing the formula $x\mr \bot$ by the quantifier-free formula $A(y,x)$. This idea can be smoothly expanded to extensions of $\PAomega$ with some additional axiom(s) $\Gamma$. Provided that $\HAomega+\Delta\vdash\Phi\mr \Gamma^{\negt}$, where $\Gamma^\negt$ denotes the negative translation of $\Gamma$, $\Delta$ is some set of axioms satisfying some natural closure properties with respect to $\bot$, and $\Phi$ is some closed term in the language of $\HAomega+\Delta$, then from a classical proof $\PAomega+\Gamma\vdash \forall y\exists x^\N A(y,x)$ one can extract a term $t$ in $\Phi$ such that $\HAomega+\Delta\vdash \forall y A(y,ty)$.  

Thus we have a method that allows us to extract realizers for $\Pi^0_2$ formulas from any extension $\Gamma$ of Peano arithmetic whenever we can realize the negative interpretation $\Gamma^{\negt}$ of $\Gamma$. In the remainder of this paper we develop this idea and focus on constructing terms $\Phi$ such that $\Phi\mr\Gamma^{\negt}$ where $\Gamma$ is either countable or countable dependent choice, and $\Gamma^{\negt}$ is the adapted realizability interpretation which treats $x\mr\bot$ as a new predicate in $x$. In fact, following \cite{EscOli(2012.0)} we generalise slightly and replace $\bot$ by some arbitrary formula $R$ whose type of realizers is a discrete type, emphasising the fact that $\bot$ acts as some undefined object to be realized. However, if the reader prefers they can just treat this as a relabelling and imagine $R=\bot$ throughout.

\subsection{The $J$-shift and its variants}
\label{sec-countable-real}

The axiom of countable choice is given by
\begin{equation*}\AC \; \colon \; \forall n\exists x^\rho B_n(x)\to\exists \alpha^{\N\to\rho}\forall n B_n(\alpha(n)).\end{equation*}
It is well-known that the negative translation of $\AC$,
\begin{equation*}\forall n((\exists x B^{\negt}_n(x)\to R)\to R)\to (\exists\alpha\forall n B^{\negt}_n(\alpha(n))\to R)\to R,\end{equation*}
(here with an arbitrary discretely-realized $R$ in place of $\bot$) is provable using the (trivially realized) intuitionistic axiom of choice from the simpler double-negation shift,
\begin{equation*}\DNS \; \colon \; \forall n((A(n)\to R)\to R)\to (\forall n A(n)\to R)\to R,\end{equation*}
by setting $A(n):=\exists x B^{\negt}_n(x)$. Thus a realizability interpretation of countable choice follows directly from that of $\DNS$. Note that this version of $\DNS$ for arbitrary $R$ is also called the $K$-shift in \cite{EscOli(2012.0)}. 

In order to successfully realize $\DNS$ one typically relies on a term $h$ realizing ex-falso-quodlibet in the form $\forall n(R\to A(n))$, and so in practice one must work with a restricted class of formulas $A(n)$ that admit such a realizer, such as any formula in the image of the negative translation (for which one can trivially construct such a $h$ even uniformly in $n$). This need for additional realizers and a corresonding restriction on formulas can seem slightly inelegant, and so a reformulation of $\DNS$ is given in \cite{EscOli(2012.0)} which, rather than separately assuming $R\to A(n)$, adds this positive information directly to the premise of $\DNS$, yielding
\begin{equation*}\forall n((A(n)\to R)\to A(n))\to (\forall n A(n)\to R)\to R.\end{equation*}
In \cite{EscOli(2012.0)} this is actually written in a equivalent form called the $J$-shift:
\begin{equation*}\mbox{$J$-shift} \; \colon \; \forall n((A(n)\to R)\to A(n))\to (\forall n A(n)\to R)\to \forall nA(n),\end{equation*}
and this is given a realizability interpretation using the product of selection functions, an interpretation which is valid for \emph{arbitrary} formulas $A(n)$. Then, in the case that $R\to A(n)$ is realizable, one easily reobtains an interpretation of the double negation shift:
\begin{proposition}[\cite{EscOli(2012.0)}]$J$-shift implies $\DNS$ over minimal logic, whenever $R\to A(n)$ holds.\end{proposition}
Here, we adopt the convention of \cite{EscOli(2012.0)} in adding the positive information we need directly to the premise of the double negation shift, so that our interpretation is valid for all $A(n)$, and as in \cite{EscOli(2012.0)}, and as we show in our examples in Section \ref{sec-dependent-examples}, we can always convert our realizer to one of $\DNS$ for negated formulas $A(n)$, which in turn is sufficient to realize the axiom of countable choice.

For notational reasons we interpret a pair of syntactically more flexible variants of the $J$-shift, designed to match the family of realizers we construct.

\begin{definition}We define the $J^\ast_i$-shifts for $i=1,2$ by
\begin{equation*}\begin{aligned}\mbox{$J^\ast_1$-shift} \; &\colon \; \forall m,n((A(m)\to R)\to A(n))\to (\forall n A(n)\to R)\to R \\
\mbox{$J^\ast_2$-shift} \; &\colon \; \forall m,n((A(m)\to R)\to A(n))\to (\forall n A(n)\to R)\to \forall nA(n)\end{aligned}\end{equation*}
where $R$ has discrete realizing type.  \end{definition}
The following result confirms that our $J^\ast_i$-shift principles are nothing more than simple rephrasings of the original $J$-shift.
\begin{proposition}$\mbox{$J^\ast_1$-shift}\Leftrightarrow \mbox{$J^\ast_2$-shift}\Leftrightarrow\mbox{$J$-shift}$ over minimal logic.\end{proposition}

\begin{proof}$\mbox{$J^\ast_1$-shift}\Rightarrow \mbox{$J^\ast_2$-shift}$ follows from the observation that 
\begin{equation*}\forall m,n((A(m)\to R)\to A(n))\to (R\to\forall nA(n)),\end{equation*}
which is true because for arbitrary $n$ we have
\begin{equation*}R\to (A(n)\to R)\end{equation*}
and thus 
\begin{equation*}((A(n)\to R)\to A(n))\to (R\to A(n)).\end{equation*}
The implication $\mbox{$J^\ast_2$-shift}\Rightarrow \mbox{$J$-shift}$ follows from 
\begin{equation*}\forall n((A(n)\to R)\to A(n))\to \forall m,n((A(m)\to R)\to A(n)),\end{equation*} 
which is true because firstly
\begin{equation*}((A(m)\to R)\to A(m))\to ((A(m)\to R)\to R),\end{equation*}
and since $R\to (A(n)\to R)$ this yields
\begin{equation*}((A(m)\to R)\to A(m))\to ((A(m)\to R)\to (A(n)\to R))\end{equation*}
and so finally 
\begin{equation*}((A(n)\to R)\to A(n))\wedge((A(m)\to R)\to A(m))\to ((A(m)\to R)\to A(n))\end{equation*} 
The remaining direction $\mbox{$J$-shift}\Rightarrow \mbox{$J^\ast_1$-shift}$ is straightforward.\end{proof}

The reasons that we highlight these variants of the $J$-shift is that we want versions of the $J$-shift that are directly realized by our variants of bar recursion. As we will see, in most cases each variant will only use the premise of $J^\ast_i$-shift for $(m,n)\in I\subseteq\N\times\N$ for some $I$.

\subsection{Realizing the $J_i^\ast$-shift}
\label{sec-countable-real}

We focus on constructing a realizer for the $J_1^\ast$-shift, then a realizer of the $J^\ast_2$-shift comes out immediately  Suppose that the realizing types of $A(n)$ and $R$ are $\rho$ and $\tau$ respectively, where we assume that $\tau$ is discrete. The $J^\ast_1$-shift is realized by a term $\Phi$ of type $(\N\to \N\to  (\rho\to\tau)\to\rho)\to (\rho^\N\to\tau)\to\tau$, which, given terms $\varepsilon\colon\N\to\N\to (\rho\to\tau)\to\rho$ and $q\colon\rho^\N\to\tau$ that satisfy
\begin{equation*}\label{eqn-varepsilon-q}\begin{aligned}&\forall m,n,p^{\rho\to\tau}(\forall x^\rho(x \mr A(m)\to p(x)\mr R)\to\varepsilon_{m,n}(p)\mr A	(n))\\
&\forall\alpha^{\rho^\N}(\forall n(\alpha(n)\mr A(n))\to q(\alpha)\mr R)\end{aligned}\end{equation*}
returns a term $\Phi\varepsilon q\colon\tau$ satisfying $\Phi\varepsilon q\mr R$.

The basic idea that unites all such existing realizers of double-negation shift principles is to form an auxiliary functional $\Psi$ which performs a backward recursive loop, which builds increasingly large partial realizers of $\forall nA(n)$. 

More precisely, suppose that $\varepsilon$ and $q$ satisfy the premise of the $J_1^\ast$ shift as above, and imagine we are given a partial realizer $u\colon \ps{\rho}^\N$ which satisfies $\forall n\in\dom(u)(u(n)\mr A(n))$. Then let us somewhat informally define $\Psi^{\varepsilon,q}u:=q(u\at c(u))$, where $c$ is some as yet unspecified function on partial sequences, but the aim is that it forms a completion of the partial realizer of $u$ and that therefore $\Psi^{\varepsilon,q}u\mr R$ for all $u$. This completion $c(u)$ can be constructed by backward recursion using $\varepsilon$. For $n\notin\dom(u)$ we define
\begin{equation*}c(u)(n):=\varepsilon_{mnu,n}(\lambda x\; . \; \Psi^{\varepsilon ,q}(\update{vnu}{m}{x}))\end{equation*}  
for some index $mnu$ and partial realizer $vnu$ of $\forall nA(n)$. Again, both $mnu$ and $vnu$ are left unspecified for now, but note that if at the very least 
\begin{enumerate}[(i)]

\item $vnu$ is a partial realizer of $\forall nA(n)$ satisfying $u\sqsubseteq vnu$ and

\item $mnu\notin\dom(u)$ whenever $n\notin\dom(u)$,

\end{enumerate}
then $u\sqsubset \update{vnu}{mnu}{x}$ and $\update{vnu}{mnu}{x}$ is a partial realizer whenever $x\mr A(mnu)$. Therefore $\lambda x.\Psi^{\varepsilon,q}(\update{vnu}{mnu}{x})$ is a realizer of $A(mnu)\to R$ under the assumption that $\Psi$ is correct for all \emph{extensions} of $u$. The idea now is that we can admit this assumption as a backward induction hypothesis, and so by backward induction we can prove that $\Psi^{\varepsilon,q}u\mr R$ for all $u$. Then setting $\Phi\varepsilon q:=\Psi^{\varepsilon,q}\emptyset$ gives us a realizer for the $J^\ast_1$-shift, since $\emptyset$ is trivially a partial realizer of $\forall n A(n)$.

What remains is to formalise this idea and make some sensible choice of $mnu$ and $vnu$ satisfying (i) and (ii) above. The most natural might be to set $mnu=n$ and $vnu=u$ - in this case, as we will see below, is precisely the idea behind the Berardi-Bezem-Coquand functional of \cite{BBC(1998.0)}. However, more intricate choices lead to other realizers, including modified bar recursion and the product of selection functions. Our aim now is to make this intuition precise, and provide a sufficiently rich general construction of $m$ and $v$ which captures all of these realizers and much more.

\begin{proposition}\label{prop-main-AC}Suppose that $\rho$ and $\tau$ are the realizing types of $A(n)$ and $R$, with $\tau$ discrete, and that were are given a computable relation $\prec\colon\ps{\rho}^\N\to (\N\times\N\to\Bool)$ and an index $m\colon\N\times \ps{\rho}^\N\to\N$ such that
\begin{enumerate}[(i)]

\item\label{item-AC-conda} $\prec_{u}$ is well-founded,

\item\label{item-AC-condb} $n\notin\dom(u)\to m{nu}\notin\dom(u)\cup\{k\; | \; k\prec_u n\}$

\end{enumerate} 
for all $u$ and $n$. Then there is a term $\Psi_{(\prec,m)}^{\varepsilon ,q}\colon\ps{\rho}^\N\to\tau$ with parameters of type $\varepsilon\colon\N\to\N\to (\rho\to\tau)\to\rho$ and $q\colon\rho^\N\to\tau$ which is primitive recursively definable in $\BR+\lambda u.\wR_{\prec_{u}}$, and satisfies the recursive equation
\begin{equation*}\Psi^{\varepsilon,q}(u)=q(\underbrace{u\at\lambda n\; . \; \varepsilon_{m{nu},n}(\lambda x\;. \; \Psi^{\varepsilon,q}(u\at\update{\initSegsss{\alpha_u}{n}}{m{nu}}{x})}_{\alpha_u})),\end{equation*} 
where $\alpha_u\colon\rho^\NN$ denotes the argument of $q$ as indicated above, and $\initSegsss{\alpha_u}{n}\colon\ps{\rho}^\NN$ is defined to be $$\lambda k.(\alpha_u(k)\mbox{ if $k\prec_{u} n$ else $\bot$}).$$\end{proposition}

\begin{remark}\label{rem-wR}Technically speaking the term $\lambda u.\wR_{\prec_{u}}$ is not properly defined in $\EHAomega$ - we are simply assuming here that there exists a function $F$ definable in $\EHAomega$ such that for all $u$, $F(u)$ satisfies the defining equation of $\wR_{\prec_{u}}$ i.e. well-founded recursion over $\prec_u$ as considered in Section \ref{sec-open-recursion}. However, in all of the concrete examples we consider, $\prec$ will not depend on $u$ and $\wR_{\prec}$ will always be trivially definable in $\EHAomega$, so this rather casual definition will not be problematic. \end{remark}

Proposition \ref{prop-main-AC} above is a special case of Proposition \ref{prop-main-DC} in the next section, whose proof can be found in the appendix, and so we omit a proof of Proposition \ref{prop-main-AC} here. However, on an informal level, it is not too difficult to see that $\Psi$ is well-defined. First we note that the definition of $\initSegsss{\alpha_u}{n}$ is not circular, since it is only used to define $\alpha_u(n)$, and thus $\alpha_u$ as a whole is constructed using the well-founded recursor $\wR_{\prec_u}$. Then one observes that the whole expression is a well-defined backward recursive functional since to compute $\Psi(u)$ we only call $\Psi$ on arguments of the form $u\at \update{\initSegsss{\alpha_u}{n}}{m}{x}$ for $m\notin\dom(u)$ by condition (\ref{item-AC-condb}), which are always strict extensions of $u$.

\begin{theorem}\label{thm-main-AC}Suppose that $\Psi_{(\prec,m)}$ is defined as in Proposition \ref{prop-main-AC} for $\prec$ and $m$ satisfying conditions (\ref{item-AC-conda}) and (\ref{item-AC-condb}). Then the term $\Phi_{(\prec,m)}:=\lambda \varepsilon,q\; . \; \Psi_{(\prec,m)}^{\varepsilon,q}(\emptyset)$ realizes the $J^\ast_1$-shift, provably in $\EHAomega+\CONT+\backI+\BR+(\wI_{\prec_{u}})+\lambda u.\wR_{\prec_u}$. \end{theorem}

\begin{remark}\label{rem-wI}Here $(\wI_{\prec_{u}})$ denotes the collection of well-founded induction schemata over the well-founded relations $\prec_u$. As with the corresponding modes of recursion discussed in Remark \ref{rem-wR}, in practise $\prec_u$ will typically not depend on $u$ and $\wI_{\prec}$ will be easily provable in $\EHAomega$.\end{remark}

\begin{proof}Assume that $\varepsilon$ and $q$ realize the premise of the $J^\ast_1$-shift. We prove that $\Phi_{(\prec,m)}\varepsilon q\mr R$ using a main backward induction and an auxiliary well-founded induction. Let us define
\begin{equation*}B(u):\equiv \forall n\in\dom(u)(u(n)\mr A(n))\to\Psi^{\varepsilon,q}(u)\mr R.\end{equation*}
This is an open formula in the sense of Proposition \ref{prop-openp} since by $\CONT$ we have $$\Psi(u)\mr R\leftrightarrow \exists n\forall w\Psi(\initSeg{u}{n}\at w)\mr R.$$ Now, to prove the backward induction step for $B$, assume that $u$ is a partial realizer of $\forall n A(n)$ (i.e. the premise of $B(u)$ holds) and suppose that $B(v)$ holds for all $v\sqsupset u$. We want to show that $\Psi^{\varepsilon,q}(u)\mr R$.

We do this by first proving using $\wI_{\prec_u}$ that $\alpha_u$ (as defined in Proposition \ref{prop-main-AC}) realizes $\forall n A(n)$. Fix $n$ and assume as an auxiliary induction hypothesis that $\alpha_u(k)\mr A(k)$ for all $k\prec_u n$. If $n\in\dom(u)$ we trivially have $\alpha_u(n)=u(n)\mr A(n)$, so assume that $n\notin\dom(u)$. In this case, first observe that
\begin{equation*}x\mr A(mnu)\to \forall i\in\dom(v)(v(i)\mr A(i))\mbox{\; \; and \; \;}\forall i\in\dom(v)(v(i)\mr A(i)) \to \Psi^{\varepsilon,q}(v)\mr R\end{equation*}
for $v:=u\at\update{\initSegsss{\alpha_u}{n}}{mnu}{x}$. The first step is clear by the auxiliary induction hypothesis, which implies that $u\at\initSegsss{\alpha_u}{n}$ is a partial realizer of $\forall nA(n)$, while the second step follows from the main hypothesis $B(v)$, since by condition (\ref{item-AC-condb}) we know that $n\notin\dom(u)$ implies that $mnu\notin\dom(u)$ and thus $v\sqsupset u$. Putting this together we see that $\lambda x.\Psi(u\at\update{\initSegsss{\alpha_u}{n}}{mnu}{x})\mr (A(mnu)\to R)$ and thus 
\begin{equation*}\alpha_u(n)=\varepsilon_{mnu,n}(\lambda x.\Psi(u\at\update{\initSegsss{\alpha_u}{n}}{mnu}{x}))\mr A(n).\end{equation*} 
by correctness of $\varepsilon$. This completes the auxiliary well-founded induction, giving us $\forall n(\alpha_u(n)\mr A(n))$, and therefore $\Psi(u)=q(\alpha_u)\mr R$, which completes the main backward induction step. Finally, then, we obtain $\forall u B(u)$ by $\backI$, and so in particular by $B(\emptyset)$ we have $\Psi^{\varepsilon,q}(\emptyset)\mr R$, which completes the proof.\end{proof}

\begin{corollary}\label{cor-main-AC}There is a term $\tilde\Phi_{({\prec},{m})}$ primitive recursive in $\BR+\lambda u.\wR_{\prec_{u}}$ which realizes the $J^\ast_2$-shift, provably in $\EHAomega+\CONT+\backI+\BR+(\wI_{\prec_{u}})+\lambda u.\wR_{\prec_u}$.\end{corollary}

\begin{proof}Keeping all the notation of Theorem \ref{thm-main-AC}, define $\tilde\Psi^{\varepsilon,q}(u)=\alpha_u$ so that it satisfies the recursive equation
\begin{equation*}\tilde\Psi^{\varepsilon,q}(u)=u\at\lambda n\; . \; \varepsilon_{m{nu},n}(\lambda x\;. \; q(\tilde\Psi^{\varepsilon,q}(u\at\update{\initSegsss{\alpha_u}{mnu}}{m{nu}}{x}))).\end{equation*}
Define $\tilde\Phi\varepsilon q:=\tilde\Psi^{\varepsilon,q}(\emptyset)$. Then it follows immediately from $\forall u B(u)$ and $\prec_\emptyset$-induction, as in the proof of Theorem \ref{thm-main-AC}, that $\tilde\Phi\varepsilon q\mr \forall nA(n)$. \end{proof}

Let us now briefly consider some specific instantiations of the parameters of Theorem \ref{thm-main-AC} (a more detailed discussion will be given in Section \ref{sec-dependent-examples}). Firstly, setting $\Psi_0=\Psi_{(\prec,m)}$ in the simple case that $\prec_u=\emptyset$ and $mnu=n$ for all $u$, we have $\initSegsss{\alpha_u}{n}=\emptyset$ and, using the abbreviation $\varepsilon_n$ for $\varepsilon_{n,n}$ our realizer becomes
\begin{equation*}\Psi_0^{\varepsilon,q}(u)=q(u\at\lambda n\; . \; \varepsilon_n(\lambda x\; . \; \Psi_0(\update{u}{n}{x})))\end{equation*}
which is nothing more than a simple variant of the Berardi-Bezem-Coquand realizer of countable choice given in \cite{BBC(1998.0)} and discussed in \cite{Berger(2004.0)}. On the other hand, suppose that we still keep $\prec_u=\emptyset$ for all $u$, but define $mnu:=\least i\leq n(i\notin\dom(u))$, where recall that $\least$ is the bounded search operator which in this case returns the least $i\leq n$ satisfying $i\notin\dom(u)$, and $n$ if no such $i$ exists. Then finite input for this variant of $\Psi$ will be of the form $\ext{s}$ for some sequences $s\colon\rho^\ast$. Defining $\Psi_1(s):=\Psi_{(\emptyset,m)}(\ext{s})$ and observing that for $n\geq |s|$ we have $mn\ext{s}=|s|$, we obtain 
\begin{equation*}\Psi^{\varepsilon,q}_1(s):=q(\ext{s}\at\lambda n\; . \; \varepsilon_{|s|,n}(\lambda x\; . \; \Psi_1(s\ast x)))\end{equation*}
which is just a (non-dependent) form of modified bar recursion. Finally, let us define $\prec_u=<$ and $mnu=n$ for all $u$. Then setting $\Psi_2(s):=\Psi_{(<,m)}(\ext{s})$, abbreviating $\varepsilon_{n,n}$ by $\varepsilon_n$ and observing that for $n\geq |s|$ we have $\ext{s}\at \initSegsss{\alpha_u}{n}=\initSeg{\alpha_u}{n}$ we obtain a realizing term satisfying
\begin{equation*}\Psi_2^{\varepsilon,q}(s)=q(\underbrace{\ext{s}\at\lambda n\; . \; \varepsilon_n(\lambda x\; . \; \Psi_2(\initSeg{\alpha_u}{n}\ast x))}_{\alpha_u}).\end{equation*}
The corresponding variant $\tilde\Psi_2$ which realizes $J^\ast_2$-shift is exactly the simple implicit product of selection functions of \cite{EscOli(2010.1)}. 

Thus three completely different modified realizability intepretations of countable choice appear as simple instances of Theorem \ref{thm-main-AC}. Moreover, in each instance we only require a restricted form of backward induction and well-founded recursion which corresponds exactly to the soundness proofs used in the original papers: for $\Psi_0$ Theorem \ref{thm-main-AC} is reduced to the proof of the double negation shift using update recursion given in \cite{Berger(2004.0)}, while $\Psi_1$ and $\Psi_2$ require backward induction relativised to downward closed partial functions, which is entirely equivalent to the variants of bar induction used to prove their correctness in \cite{BergOli(2005.0)} and \cite{EscOli(2010.1)} respectively. Thus Theorem \ref{thm-main-AC} doesn't simply provide a parametrised framework with which different realizers can be compared, but also a framework in which their correctness proofs can be viewed in a uniform way as relativisations of backward induction.

Of course the construction of such a framework is only partially motivated by the desire to compare existing interpretations. Theorem \ref{thm-main-AC} generalises existing work in that one can use an arbitrary parameters to define new realizers of the $J^\ast_i$-shift that are automatically correct, giving an additional level of flexibility and power when it comes to extracting computational content from proofs in practise.

However, we do not discuss this in any more detail here, instead proceeding straight to the generalisation of Theorem \ref{thm-main-AC} to full dependent choice.

\section{A computational interpretation of dependent choice}
\label{sec-dependent}

We now give a parametrised realizability interpretation to the principle of countable dependent choice. Our formulation of dependent choice will be slightly more general that the usual sequential variants treated in e.g. \cite{BergOli(2005.0),EscOli(2012.0),Seisenberger(2008.0)}, in the sense that we parametrise the principle itself by a decidable well-founded relation $\lhd$ on $\N$ which dictates the underlying dependency of the choice sequence.

To be more precise, given a decidable strict well-founded partial order $\lhd$ let us extend our type system with types $\dc{\rho}{\lhd}$ which represent the set $\bigcup_{n\in\N} \dc{\rho}{\lhd_n}$ where
\begin{equation*}\dc{\rho}{\lhd_n}:\equiv \{m \; | \; m\lhd n\}\to\rho. \end{equation*}
We tacitly assume that the types $\dc{\rho}{\lhd}$ can be smoothly incorporated into our system, and come equipped with a length function $|\cdot|\colon\dc{\rho}{\lhd}\to \N$ returning for each $t\colon\dc{\rho}{\lhd}$ a unique index $|t|$ such that $t\in \dc{\rho}{\lhd_{|t|}}$. For $\lhd=<$, the type $\dc{\rho}{<}$ is isomorphic to the type $\rho^\ast$ of finite sequences over $\rho$, objects of type
\begin{equation*}\{m \; | \; m<n\}\to\rho\end{equation*}
representing finite sequences of length $n$. In fact $\dc{\rho}{\lhd}$ is essentially a generalisation of the finite sequence type to arbitrary $\lhd$-closed partial functions (which need not have finite domain, though). Note that when $\lhd=\emptyset$ is the empty relation, the type $\dc{\rho}{\emptyset}$ is isomorphic to $\N$, and as we will see, in this case our parametrised dependent choice principle collapses to normal countable choice.

Now, the principle of $\oDC{\lhd}$ is given by the schema
\begin{equation*}\forall s^{\dc{\rho}{\lhd}}(\forall i\lhd |s|A_i(\initSegss{s}{i},s(i))\to \exists x^\rho A_{|s|}(s,x))\to\exists \alpha^{\N\to\rho}\forall n A_n(\initSegss{\alpha}{n},\alpha(n))\end{equation*}
where $\initSegss{\alpha}{n}\colon \dc{\rho}{\lhd n}$ is just the $\lhd$-initial segment of $\alpha$ i.e. $\lambda m\lhd n.\alpha(m)$, and analogously for $s$ (note that $\initSegss{s}{i}$ is well-defined for $i\lhd |s|$ by transitivity of $\lhd$). Note that $\oDC{\lhd}$ follows from the full axiom of choice together with well-founded recursion over $\lhd$ as follows: by classical logic and full choice we have
\begin{equation*}\forall s(\forall i\lhd |s|A_i(\initSegss{s}{i},s(i))\to \exists x A_{|s|}(s,x))\to \exists \Theta\forall s(\forall i\lhd |s|A_i(\initSegss{s}{i},s(i))\to  A_{|s|}(s,\Theta s)).\end{equation*}
Now, recursively defining $\alpha(n):=\Theta(\initSegss{\alpha}{n})$, we prove $\forall n A_n(\initSegss{\alpha}{n},\alpha(n))$ by $\wI_{\lhd}$, since from the assumption that $\forall i\lhd n A_i(\initSegss{\alpha}{i},\alpha(i))$ we obtain $A_n(\initSegss{\alpha}{n},\alpha(n))$ using that fact that $\initSegss{\initSegss{\alpha}{n}}{i}=\initSegss{\alpha}{i}$ for $i\lhd n$.

In Section \ref{sec-dependent-examples} below we discuss more well-known instances dependent choice, including its canonical formulation as
\begin{equation*}\forall n,x^X\exists y A_n(x,y)\to\exists f\forall n A_n(f(n),f(n+1)),\end{equation*}
which is easily provable from $\oDC{<}$. However, here we take advantage of the fact that our setting allows us to interpret dependent choice in the non-standard, but very general form $\oDC{\lhd}$.

\subsection{Realizing $\oDC{\lhd}^{J^\ast_i}$}
\label{sec-dependent-real}

As in the previous section, we realize a positive form of the negative translation of $\oDC{\lhd}$ which is somewhat analogous to the dependent $J$-shift of \cite{EscOli(2012.0)}. As with the previous section, we can easily derive a realizer for the standard negative translation of dependent choice principles from a realizer of our shift principles - and we illustrate this in Section \ref{sec-dependent-examples}.

\begin{definition}\label{defn-Nshift-dep} We define the translated principle $\oDC{\lhd}^{J^\ast_i}$ by
\begin{equation*}\begin{aligned}\oDC{\lhd}^{J^\ast_1} \; &\colon \; \begin{cases}\forall s,r(\forall i\lhd |s| A_i(\initSegss{s}{i},s(i))\to (\exists x A_{|s|}(s,x)\to R)\to \exists x A_{|r|}(r,x)) \\ 
\to(\exists \alpha\forall n A_n(\initSegss{\alpha}{n},\alpha(n))\to R)\to R\end{cases}\\
\oDC{\lhd}^{J^\ast_2} \; &\colon \; \begin{cases}\forall s,r(\forall i\lhd |s| A_i(\initSegss{s}{i},s(i))\to (\exists x A_{|s|}(s,x)\to R)\to \exists x A_{|r|}(r,x)) \\ 
\to(\exists \alpha\forall n A_n(\initSegss{\alpha}{n},\alpha(n))\to R)\to \exists \alpha\forall n A_n(\initSegss{\alpha}{n},\alpha(n))\end{cases}\end{aligned}\end{equation*}
where $A$ is arbitrary formula over $\N\times\dc{\rho}{\lhd}\times \rho$ and the realizing type of $R$ is restricted to being discrete.\end{definition}

Again, these variants are nothing more than convenient syntactical rephrasings of the dependent shift principle given in \cite{EscOli(2012.0)}, this time extended to arbitrary partial orderings $\lhd$.

Now, if $\sigma$ and $\tau$ are the realizing types of $A_n(s,x)$ and $R$ respectively, then $\oDC{\lhd}^{J^\ast_1}$ is realized by a term $\Phi$ of type $(\dc{(\rho\times\sigma)}{\lhd}\to\dc{\rho}{\lhd}\to (\rho\times\sigma\to\tau)\to\rho\times\sigma)\to ((\rho\times\sigma)^\N\to\tau)\to\tau$ which given input $\varepsilon\colon\dc{(\rho\times\sigma)}{\lhd}\to\dc{\rho}{\lhd}\to (\rho\times\sigma\to\tau)\to\rho\times\sigma$ and $q\colon (\rho\times\sigma)^\N\to\tau$ that satisfies
\begin{equation*}\begin{aligned}&\forall s^{\dc{(\rho\times\sigma)}{\lhd}},r^{\dc{\rho}{\lhd}},p^{\rho\times\sigma\to\tau}\begin{cases}\forall i\lhd |s|(s(i)_1\mr A_i(\initSegss{s_0}{i},s(i)_0))\to \\ (\forall x^{\rho\times\sigma}(x_1\mr A_{|s|}(s_0,x_0)\to p(x)\mr R)\to \varepsilon_{s,r}(p)_1\mr A_{|r|}(r,\varepsilon_{s,r}(p)_0))\end{cases}\\
&\forall\alpha^{(\rho\times\sigma)^\N}(\forall n \; \alpha(n)_1\mr A_n(\initSegss{\alpha_0}{n},\alpha(n)_0)\to q(\alpha)\mr R)\end{aligned}\end{equation*}
returns a term $\Phi\varepsilon q\colon\tau$ satisfying $\Phi\varepsilon q\mr R$. In the formulae above and results that follows, $s_0\colon\dc{\rho}{\lhd}$ denotes the first projection of the term $s\colon\dc{(\rho\times\sigma)}{\lhd}$ i.e. $s_0(m):=s(m)_0$, and similarly for infinite sequences $\alpha$. As in the previous section, we construct a family of realizing terms which follow the same basic principle of backward recursion on partial functions, and which give a computational interpretation to $\oDC{\lhd}^{J^\ast_1}$. However, this time our choices of $\prec_u$ and $m{nu}$ require some additional conditions to ensure they are now compatible with $\lhd$.

\begin{proposition}\label{prop-main-DC}Suppose that $\sigma$ and $\tau$ are the realizing types of $A(n)$ and $R$, with $\tau$ discrete, and that $\prec\colon\wideps{(\rho\times\sigma)}^\N\to (\N\times\N\to\Bool)$ and $m\colon\N\times \wideps{(\rho\times\sigma)}^\N\to\N$ are such that
\begin{enumerate}[(i)]

\item\label{item-DC-conda} $\forall k(k\prec_u n\vee k\lhd n\vee k\lhd mnu\to k\prec'_u n)$ for some well-founded relation $\prec'_u$,

\item\label{item-DC-condb} $n\notin\dom(u)\to m{nu}\notin\dom(u)\cup\{k\; | \; k\prec_u n\}$,

\end{enumerate} 
for all $u$ and $n$. Then there is a term $\Psi^{\varepsilon,q}_{(\lhd,\prec,m)}\colon\overline{(\rho\times\sigma)}^\N\to\tau$ with parameters of type $\varepsilon\colon(\dc{(\rho\times\sigma)}{\lhd}\to\dc{\rho}{\lhd}\to (\rho\times\sigma\to\tau)\to\rho\times\sigma)$ and $q\colon((\rho\times\sigma)^\N\to\tau)$ which is primitive recursively definable in $\BR+\lambda u.\wR{\prec'_u}$, and satisfies the recursive equation 
\begin{equation*}\Psi^{\varepsilon,q}(u)=q(\underbrace{u\at\lambda n\;. \; \varepsilon_{\initSegss{\alpha}{mnu},\initSegss{\alpha_0}{n}}(\lambda x\; . \; \Psi^{\varepsilon,q}(u\at \update{\initSegsss{\alpha_u}{n}}{mnu}{x}))}_{\alpha_u})\end{equation*}
where $\alpha_u\colon (\rho\times\sigma)^\N$ denotes the argument of $q$ as indicated above and $\initSegsss{\alpha_u}{n}\colon\wideps{(\rho\times\sigma)}^\N$ is defined to be $$\lambda k.(\alpha_u(k)\mbox{ if $k\prec_u n$ else $\bot$}).$$\end{proposition}

Regarding the somewhat casual definition of the term $\lambda u.\wR{\prec'_u}$ see Remark \ref{rem-wR} - again, in all concrete cases discussed later $\prec_u$ will not even depend on $u$ and $\wR{\prec'_u}$ will be trivially definable in $\EHAomega$. As with Proposition \ref{prop-main-AC}, on an intuitive level it is not too hard to see that $\Psi$ is well-defined. Firstly $\alpha_u$ is well-defined by $\prec'_u$-recursion since in order to compute $\alpha_u(n)$ we require $\alpha_u(k)$ for either $k\prec_u n$, $k\lhd n$ or $k\lhd mnu$, which by condition (\ref{item-DC-conda}) implies that $k\prec'_u n$. Then by (\ref{item-DC-condb}) for $n\notin\dom(u)$ we have $mnu\notin\dom(u)$ and thus $u\sqsubset u\at \update{\initSegsss{\alpha}{n}}{mnu}{x}$ and so recursive calls on $\Psi$ are always made on strict extensions $u$. A formal proof that $\Psi$ is definable in $\BR+\lambda u.\wR{\prec_u}$ is given in Appendix \ref{sec-app}.

\begin{theorem}\label{thm-main-DC}Suppose that $\Psi_{(\lhd,\prec,m)}$ is defined as in Proposition \ref{prop-main-DC} for $\lhd$, $\prec$ and $m$ satisfying conditions (\ref{item-DC-conda}) and (\ref{item-DC-condb}), and additionally
\begin{enumerate}[(i)]

\item[(iii)]\label{item-DC-condc} if $\dom(u)$ is $\lhd$-closed then so are $\dom(u)\cup\{k\; | \; k\prec_u n\}$ and $\dom(u)\cup\{k\; | \; k\prec_u n\}\cup \{m{nu}\}$,

\end{enumerate}
Then the term $\Phi_{(\lhd,\prec,m)}:=\lambda\varepsilon,q\;.\; \Psi^{\varepsilon,q}_{(\lhd,\prec,m)}(\emptyset)$ realizes $\oDC{\lhd}^{J^\ast_1}$ provably in $\T+\CONT+\backI+\BR+(\wI_{\prec'_{u}})+\lambda u.\wR_{\prec'_u}$.\end{theorem}


\begin{remark}While the proof of this Theorem is fairly lengthy, it is little more than a straightforward variant of the much simpler proof of Theorem \ref{thm-main-AC}. The additional technical details are simply required to deal with the dependency $\lhd$.\end{remark}

\begin{remark}All terms $s\colon\dc{\rho}{\lhd}$, when embedded as partial functions, are $\lhd$-closed by transitivity of $\lhd$, but the converse is not necessarily true.\end{remark}

\begin{proof}Assume that $\varepsilon$ and $q$ realize the premise of $\oDC{\lhd}^{J^\ast_1}$. Analogous to the proof of Theorem \ref{thm-main-AC}, we use a main backward induction along with an auxiliary well-founded induction, this time on the slightly more complex formula 
\begin{equation*}B(u):\equiv \bar{A}(u)\to\Psi^{\varepsilon,q}(u)\mr R.\end{equation*}
for
\begin{equation*}\bar{A}(u):\equiv\forall n\in\dom(u)(\forall k\lhd n(k\in\dom(u))\wedge u(n)_1\mr A(\initSegss{{u_0}}{n},u(n)_0)),\end{equation*}
where now as our premise we require that $u$ is a $\lhd$-closed partial realizer of the dependent choice sequence. As before, $B(u)$ is open in the sense of Proposition \ref{prop-openp} by $\CONT$. To prove the backward induction step, let's assume that $\bar{A}(u)$ holds and $B(v)$ is true for all $v\sqsupset u$ and try to derive $\Psi(u)\mr R$.

This time we use $\prec_u'$ induction to show that $\forall n C(n)$ where
\begin{equation*}C(n):\equiv \alpha_u(n)_1\mr A_n(\initSegss{(\alpha_u)_0}{n},\alpha_u(n)_0).\end{equation*}
Then we are done since this would imply that $\Psi(u)=q(\alpha_u)\mr R$. So fix $n$ and assume that $\forall k\prec_u' n\; C(k)$. There are two cases to consider. If $n\in\dom(u)$ and then by $\lhd$-closedness of $u$, $C(n)$ becomes $u(n)_1\mr A_n(\initSegss{u_0}{n},u(n)_0)$ which is true by assumption. Otherwise, if $n\notin\dom(u)$ then $C(n)$ becomes
\begin{equation*}(\ast) \ \ \varepsilon_{\initSegss{\alpha}{mnu},\initSegss{\alpha_0}{n}}(p)_1\mr A_n(\initSegss{(\alpha_u)_0}{n},\varepsilon_{\initSegss{\alpha}{mnu},\initSegss{\alpha_0}{n}}(p)_0)\end{equation*}
where
\begin{equation*}p:=\lambda x\; . \; \Psi^{\varepsilon,q}(u\at \update{\initSegsss{\alpha_u}{n}}{mnu}{x}).\end{equation*}
We now prove that
\begin{equation*}\begin{aligned}&x_1\mr A_{mnu}(\initSegss{(\alpha_u)_0}{mnu},x_0)\to \bar{A}(u\at \update{\initSegsss{\alpha_u}{n}}{mnu}{x})\mbox{\;  and \;}\\
&\bar{A}(u\at \update{\initSegsss{\alpha_u}{n}}{mnu}{x})\to\Psi(u\at\update{\initSegsss{\alpha_u}{n}}{mnu}{x})\mr R.\end{aligned}\end{equation*}
The second implication follows from the main backward induction hypothesis since $u\sqsubset u\at\update{\initSegsss{\alpha_u}{n}}{mnu}{x}$ by (ii). To prove the first implication, let's abbreviate $v:=u\at \update{\initSegsss{\alpha_u}{n}}{mnu}{x}$. Then $\dom(v)=\dom(u)\cup\{k\; | \; k\prec_u n\}\cup\{mnu\}$, and since by $\bar{A}(u)$ we know that $\dom(u)$ is $\lhd$-closed, so is $\dom(v)$ by condition (iii). Thus to obtain $\bar A(v)$ it remains to show that
\begin{equation*}(\ast\ast) \ \ \forall k\in\dom(v)(v(k)_1\mr A(\initSegss{v_0}{k},v(k)_0)).\end{equation*}
There are three possibilities. If $k\in\dom(u)$ then $(\ast\ast)$ follows from $\bar{A}(u)$, while if $k=mnu$ then $(\ast\ast)$ becomes $x_1\mr A(\initSegss{(\alpha_u)_0}{mnu},x_0)$ by (iii) which is our premise. Otherwise, for $k\prec_u n$, $(\ast\ast)$ becomes $\alpha_u(k)_1\mr A(\initSegss{(\alpha_u)_0}{k},\alpha_u(k)_0)$ by condition (iii), and this we know is true by the $\prec'_u$-induction hypothesis.

Therefore, putting everything together we obtain
\begin{equation*}x_1\mr A_{mnu}(\initSegss{(\alpha_u)_0}{mnu},x_0)\to p(x)\mr R. \end{equation*}
Since in addition we have $\forall i\lhd mnu(\alpha_u(i)_1\mr A_i(\initSegss{(\alpha_u)_0}{i},\alpha_u(i)_0)$, a fact which also follows from the $\prec'_u$-induction hypothesis along with condition (i), then by correctness of $\varepsilon$ we obtain $(\ast)$.

Thus we have shown that $\forall k\prec_u' n\; C(k)\to C(n)$ and therefore $\forall n C(n)$ follows by induction. Since this in turn implies $\Psi(u)\mr R$, we have shown $\forall v\sqsupset u B(v)\to B(u)$, and thus $\forall u B(u)$ by backward induction. Finally, then, since $\bar{A}(\emptyset)$ is trivially satisfied, $B(\emptyset)$ implies $\Psi^{\varepsilon,q}(\emptyset)\mr R$, which completes the proof.    \end{proof}

\begin{corollary}\label{cor-main-DC}There is a term $\tilde\Phi_{(\lhd,{\prec},{m})}$ primitive recursive in $\BR+\lambda u.\wR_{\prec_{u}}$ which realizes the $\oDC{\lhd}^{J^\ast_2}$, provably in $\EHAomega+\CONT+\backI+\BR+(\wI_{\prec_{u}})+\lambda u.\wR_{\prec_u}$.\end{corollary}

\begin{proof}Just as in the previous section, and keeping the notation of Theorem \ref{thm-main-DC}, define $\tilde\Psi^{\varepsilon,q}(u)=\alpha_u$ so that it satisfies the recursive equation
\begin{equation*}\tilde\Psi^{\varepsilon,q}(u)=u\at\lambda n\;. \; \varepsilon_{\initSegss{\alpha}{mnu},\initSegss{\alpha_0}{n}}(\lambda x\; . \; q(\tilde\Psi^{\varepsilon,q}(u\at \update{\initSegsss{\alpha_u}{n}}{mnu}{x}))).\end{equation*}
Define $\tilde\Phi\varepsilon q:=\tilde\Psi^{\varepsilon,q}(\emptyset)$. Then it follows immediately from $\forall u B(u)$ and $\prec_\emptyset'$-induction that $\tilde\Phi\varepsilon q=\alpha_{\emptyset}\mr \exists\alpha\forall nA(n)$ since $\forall n\; \alpha_\emptyset(n)_1\mr A_n(\initSegss{(\alpha_{\emptyset})_0}{n},\alpha_\emptyset(n)_1)$. \end{proof}

\subsection{Examples}
\label{sec-dependent-examples}

We now show that essentially all of the solutions to the modified realizability interpretation choice principles given across e.g. \cite{BBC(1998.0),BergOli(2005.0),EscOli(2012.0),Seisenberger(2008.0)} appear as special cases of Theorem \ref{thm-main-DC}, given suitable instantiations of $\lhd$, $\prec$, $m$ and the choice formula $A_n(s,x)$.

\subsubsection{The Berardi-Bezem-Coquand functional ($\lhd=\prec_u=\emptyset$)}
\label{sec-dependent-examples-BBC}

It is easy to see that the principle $\oDC{\emptyset}$ is just countable choice, since $\dc{\rho}{\emptyset_n}$ is just a singleton object $\{\bullet\}$ indexed by $n$, and so $\dc{\rho}{\emptyset}$ is isomorphic to $\N$. Setting $A_n(\bullet,x):=B_n(x)$ we obtain $\AC$ as defined in Section \ref{sec-countable}. In fact, Theorem \ref{thm-main-DC} completely reduces to Theorem \ref{thm-main-AC} for the formula $A(n):=\exists x B_n(x)$ once we eliminate $\lhd$. Therefore as expected there is a direct corresondence between the realizers of the two theorems in this case.

For any function $m$ satisfying $n\notin\dom(u)\to mnu\notin\dom(u)$, we can define a generalised version of the Berardi-Bezem-Coquand functional as $\BBC_{(m)}:=\Psi_{(\emptyset,\emptyset,m)}$, which satisfies the defining equation
\begin{equation*}\BBC_{(m)}^{\varepsilon,q}(u)=q(u\at\lambda n\; . \; \varepsilon_{mnu,n}(\lambda x\; . \; \BBC_{(m)}^{\varepsilon,q}(\update{u}{mnu}{x}))). \end{equation*}
Regardless of the choice of $m$, this functional always gives a computational interpretation to countable choice. Now suppose that we move back into the more conventional setting of the double negation translation of countable choice, setting $R=\bot$ and assuming $B_n(x)$ is a negated formula, so that there exists a term $h$ satisfying $\forall n,x\; (h\mr (\bot\to B_n(x))($. Then if $\phi$ satisfies
\begin{equation*}\phi\mr\forall m((\exists x B_m(x)\to \bot)\to\bot)\end{equation*}
then $\varepsilon^\phi_{m,n}(p):=_{\rho\times\sigma}\pair{0_\rho,h(\phi_m(p))}$ realizes the premise of $\oDC{\emptyset}^{J^\ast_1}$, and thus defining $\BBC_{(m),1}^{\phi,q}:=\BBC	_{(m)}^{\varepsilon^\phi,q}$ - which satisfies the equation
\begin{equation*}\BBC_{(m),1}^{\phi,q}(u)=q(u\at\lambda n\; . \; \pair{0,h(\phi_{mnu}(\lambda x\; . \; \BBC_{(m),1}^{\varepsilon,q}(\update{u}{mnu}{x})))}). \end{equation*}
- we obtain a term which realizes the negative translation of $\AC$. In particular, setting $mnu=n$ we obtain
\begin{equation*}\BBC_{2}^{\phi,q}(u)=q(u\at\lambda n\; . \; \pair{0,h(\phi_{n}(\lambda x\; . \; \BBC_{2}^{\varepsilon,q}(\update{u}{n}{x})))}) \end{equation*}
which is just the BBC functional of \cite{BBC(1998.0)} (more precisely, the variant of the Berardi-Bezem-Coquand functional for input with arbitrary domain considered in \cite{Berger(2004.0)}).

\subsubsection{Modified bar recursion ($\prec_u=\emptyset$)}
\label{sec-dependent-examples-MBR}

Suppose that we retain the simplification $\prec_u=\emptyset$, but now allow $\lhd$ to range over arbitrary decidable partial orders. Then we obtain a realizer which makes recursive calls over updates of its input, just like the BBC functional, but now $m$ must not only satisfy conditions (\ref{item-DC-conda}) and (\ref{item-DC-condb}) but also (\ref{item-DC-condc}):
\begin{equation*}\mbox{$u$ is $\lhd$-closed }\to\mbox{ $u\cup\{mnu\}$ is $\lhd$-closed}\end{equation*}
In this case, the term $\MBR^{\varepsilon,q}_{(\lhd,m)}:=\Psi^{\varepsilon,q}_{(\lhd,\emptyset,m)}$, which has defining equation
\begin{equation*}\MBR^{\varepsilon,q}_{(\lhd,m)}(u)=q(u\at\lambda n\;. \; \varepsilon_{\initSegss{\alpha}{mnu},\initSegss{\alpha_0}{n}}(\lambda x\; . \; \MBR^{\varepsilon,q}_{(\lhd,m)}(\update{u}{mnu}{x})))\end{equation*}
forms a realizer for $\oDC{\lhd}$. This can be viewed as a generalisation of modified bar recursion as first defined in \cite{BergOli(2005.0)}. All existing variants of $\MBR$ occur when $\lhd$ is the usual ordering on $\N$, but $\MBR$ is perfectly well-defined in more unusual cases. For instance, suppose that we have a bijective encoding $c\colon\N\to\Bool^\ast$, and that
\begin{equation*}m\lhd n:=\mbox{$c(m)$ is a proper prefix of $c(n)$}.\end{equation*}
Then $\lhd$-closed partial functions are precisely partial functions whose domain is a binary tree, and in this case there are many valid choices for $m$, a canonical one being
\begin{equation*}mnu:=\mbox{$n$ if $n\in\dom(u)$ else $i$ where $c(i)$ is the least prefix of $c(n)$ not in $\dom(u)$}. \end{equation*}
This variant of $\MBR$ yields an intuitive realizer for dependent choice over binary trees.

Now, suppose that we do indeed have $\lhd=<$. Then $<$-closed partial functions are either total or of the form $\ext{s}$ where $s$ is a finite sequence. Observing that in order to evaluate $\MBR(\emptyset)$ we can restrict ourselves to input with finite domain we can redefine our realizer in this case as $\MBR_1^{\varepsilon,q}(s^{\rho^\ast}):=\MBR_{(<,m)}^{\varepsilon,q}(\ext{s})$, setting $mnu:=\least i\leq n(i\notin\dom(u))$. Clearly such an $m$ satisfies (\ref{item-DC-condc}), and in particular $mn\ext{s}=|s|$ for $n\geq |s|$. Therefore $\MBR_1$ has defining equation
\begin{equation*}\MBR^{\varepsilon,q}_{1}(s)=q(\ext{s}\at\lambda n\;. \; \varepsilon_{s,\initSegss{\alpha_0}{n}}(\lambda x\; . \; \MBR^{\varepsilon,q}_{1}(s\ast x)))\end{equation*}
This directly realizes $\oDC{<}$, which is isormorphic to
\begin{equation*}\forall s^{\rho^\ast}(\forall i<|s| A_i(\initSeg{s}{i},s(i))\to\exists x A_{|s|}(s,x))\to\exists\alpha\forall n A_n(\initSeg{\alpha}{n},\alpha(n)).\end{equation*}
We can now easily rederive various concrete instances of $\MBR_1$ found in the literature which arise from setting $R=\bot$ and instantiating $A_i(s,x)$ by the correct formula. First, note that we immediately derive $\AC$ from $\oDC{<}$ by setting $A_n(s,x):=B_n(x)$, and a corresponding realizer for $\AC$ by defining $\MBR_2^{\phi,q}:=\MBR_1^{\varepsilon^\phi,q}$ for $\varepsilon_{s,r}^\phi(p)=\pair{0_\rho,h(\phi_{|s|}(p))}$, where $\phi$ and $h$ are as in Section \ref{sec-dependent-examples-BBC}. This has defining equation
\begin{equation*}\MBR^{\phi,q}_{2}(s)=q(\ext{s}\at\lambda n\;. \; \pair{0,h(\phi_{|s|}(\lambda x\; . \; \MBR^{\phi,q}_{2}(s\ast x)))})\end{equation*}
which is exactly the realizer of countable choice constructed in \cite{BergOli(2005.0)}. Note that the same realizer could have been constructed from $\BBC_{(m),1}$ for suitable $m$.

Now suppose that $A_0(s,x):=B_0(x_0,x)$ and $A_n(s,x):=B_n(s_{|s|-1},x)$ for $n>0$. Then $\oDC{<}$ immediately implies the following, standard formulation of dependent choice:
\begin{equation*}\label{eqn-DC-BergOli}\DC \; \colon \; \forall n,y\exists x B_n(y,x)\to\forall x_0\exists \alpha(\alpha(0)=x_0\wedge\forall n B_n(\alpha(n),\alpha(n+1)).\end{equation*}
The challenge for realizing $\DC$ is as follows: we must construct a realizer of $\bot$, given realizers
\begin{equation*}\begin{aligned}\phi^{\N\to\rho\to(\rho\times\sigma)\to\tau)\to\tau} &\mr \forall n,y((\exists x B_n(y,x)\to \bot)\to \bot)\\
Y^{(\rho\times\sigma)^\N\to\tau}&\mr \exists\alpha(\alpha(0)=x_0\wedge\forall nB_n(\alpha(n),\alpha(n+1)))\to \bot\end{aligned}\end{equation*}
and in addition assuming that $B$ is negated, thus guaranteeing the existence of a realizer $h$ of ex-falso quodlibet satisfying $\forall n,y,x\;  h\mr (\bot\to B_n(y,x))$. But in this case, we can easily define realizers of the premise of $\oDC{<}$ as
\begin{equation*}\begin{aligned}\varepsilon^\phi_{s,r}(p)&:=\pair{0,h(\phi_{|s|,(x_0\ast s_0)_{|s|}}(p))}\\
q^Y(\alpha)&:= Y_{x_0}(\alpha):= Y(\pair{x_0\ast\alpha_0,\alpha_1})\end{aligned}\end{equation*}
and defining $\MBR_3^{\phi,Y}:=\MBR_1^{\varepsilon^\phi,q^Y}$ yields a realizer for $\DC$ satisfying
\begin{equation*}\MBR^{\phi,Y}_{3}(s)=Y_{x_0}(\ext{s}\at\lambda n\;. \; \pair{0,h(\phi_{|s|,(x_0\ast s_0)_{|s|}}(\lambda x\; . \; \MBR^{\phi,Y}_{3}(s\ast x)))})\end{equation*}
which is this time exactly the realizer of $\DC$ given in \cite{BergOli(2005.0)}.
In an entirely analogous way, the bar recursive solution to the sequential variant of dependent choice considered in \cite{Seisenberger(2008.0)}:
\begin{equation}\label{eqn-DC-Seis}B(\pair{})\to\forall s^{\rho^\ast}(B(s)\to\exists xB(s\ast x))\to\exists \alpha\forall n B(\initSeg{\alpha}{n})\end{equation}
can be defined in terms of $\MBR_1$, since (\ref{eqn-DC-Seis}) is easily implied by $\oDC{<}$ for $A_n(s,x):=B(s\ast x)$. Given realizers
\begin{equation*}\begin{aligned}a_0^\rho&\mr B(\pair{})\\
\phi^{\rho^\ast\to\sigma\to (\rho\times\sigma\to\tau)\to\tau}&\mr \forall s(B(s)\to (\exists xB(s\ast x)\to \bot)\to \bot)\\
Y^{(\rho\times\sigma)^\N\to\tau}&\mr \exists \alpha\forall n B(\initSeg{\alpha}{n})\to \bot\end{aligned}\end{equation*}
and assuming the existence of a realizer $\forall s\; h\mr (R\to B(s))$, it is easy to see that 
\begin{equation*}\begin{aligned}\varepsilon_{s,r}^{a_0,\phi}(p)&:=\pair{0,h(\phi_{s_0,(a_0\ast s)_{|s|}}(p)}\\
q^{a_0,Y}(\alpha)&:= Y_{a_0}(\alpha):= Y(\pair{\alpha_0,G_0\ast\alpha_1})\end{aligned}\end{equation*}
realize the premise of $\oDC{<}$, and therefore the realizer we obtain is $\MBR^{a_0,\phi,Y}_4:=\MBR_1^{\varepsilon_{s,r}^{a_0,\phi},q^{a_0,Y}}$ which satisfies
\begin{equation*}\MBR^{a_0,\phi,Y}_4(s)=Y_{a_0}(\ext{s}\at\lambda n\;. \; \pair{0,h(\phi_{s_0,(a_0\ast s)_{|s|}}(\lambda x\; . \; \MBR^{a_0,\phi,Y}_4(s\ast x)))})\end{equation*}
which is exactly the term used to interpret dependent choice in \cite{Seisenberger(2008.0)} (and also extract an algorithm from Higman's lemma in \cite{Seisenberger(2003.0)}).

\subsubsection{Products of selection functions ($mnu=n$, $\lhd$ total)}
\label{sec-dependent-examples-IPS}

We finally consider the case in which the functional $m$ is defined to be the projection function $mnu=n$. Clearly, the conditions (i)-(iii) are satsified whenever $\prec_u=\lhd$, in which case the realizer simplifies to 
\begin{equation*}\Psi^{\varepsilon,q}(u)=q(\underbrace{u\at\lambda n\;. \; \varepsilon_{\initSegss{\alpha}{n}}(\lambda x\; . \; \Psi^{\varepsilon,q}(u\at \update{\initSegss{\alpha}{n}}{n}{x}))}_{\alpha_u})\end{equation*}
where we abbreviate $\varepsilon_{\initSegss{\alpha}{n}}=\varepsilon_{\initSegss{\alpha}{n},\initSegss{\alpha_0}{n}}$ and $\initSegss{\alpha}{n}$ is now treated as a partial function. Now, suppose that $\lhd$ is a \emph{total} order, and therefore constitutes an encoding of some countable ordinal $\xi$. It is always the case that we can evaluate $\Psi(\emptyset)$ by restricting the input to being $\lhd$-closed. However, when $\lhd$ is total then the set of (domain-theoretically) non-total $\lhd$-closed partial sequences is isomorphic to $\dc{\rho}{\lhd}$, since if $u$ is $\lhd$-closed then $\dom(u)=\{k \; | \; k\lhd n\}$ where $n$ is the least element of the set of undefined elements of $u$. Let us suppose that $\lhd$ comes equipped with a computable ordinal successor function $\suc\colon\N\to\N$ i.e. for each number $n$, $\suc(n)$ is the least number greater than $n$ with respect to $\lhd$. Define $\PS_{(\lhd)}^{\varepsilon,q}\colon\dc{\rho}{\lhd}\to\rho^\N$ by $\PS^{\varepsilon,q}(s):=\tilde\Psi^{\varepsilon,q}(\ext{s})$ (where $\tilde\Psi$ is the term defined in Corollary \ref{cor-main-DC}). Then $\PS$ satisfies the equation
\begin{equation*}\PS_{(\lhd)}^{\varepsilon,q}(s)=\ext{s}\at\lambda n\;. \; \varepsilon_{\initSegss{\alpha}{n}}(\lambda x\; . \; q(\PS_{(\lhd)}^{\varepsilon,q}(\initSegss{\alpha}{n}\star x))).\end{equation*}
where $\initSegss{\alpha}{n}\star x\colon \dc{\rho}{\lhd_{\suc(n)}}$ is defined by
\begin{equation*}(\initSegss{\alpha}{n}\star x)(m):=\begin{cases}\alpha(m) & \mbox{if $m\lhd n$}\\ x & \mbox{if $m=n$}.\end{cases}\end{equation*}
This functional is a generalisation of the (implicitly well-founded) product of selection functions of Escard{\'o} and Oliva to arbitrary recursive ordinals, and is not only a realizer $\oDC{\lhd}^{J^\ast_2}$, but a direct realizer of the following transfinite $J$-shift principle:
\begin{equation*}\begin{cases}\forall s(\forall i\lhd |s| A_i(\initSegss{s}{i},s(i))\to (\exists x A_{|s|}(s,x)\to R)\to \exists x A_{|s|}(s,x)) \\ 
\to(\exists \alpha\forall n A_n(\initSegss{\alpha}{n},\alpha(n))\to R)\to \exists \alpha\forall n A_n(\initSegss{\alpha}{n},\alpha(n)).\end{cases}\end{equation*}
For the particular case that $\lhd$ is the normal ordering on $\N$ and $A_n(s,x):=B_n(s\ast x)$ this becomes
\begin{equation*}\begin{cases}\forall s^{\rho^\ast}(\forall i< |s| B_i(\initSeg{s}{i+1})\to (\exists x B_{|s|}(s\ast x)\to R)\to \exists x B_{|s|}(s\ast x)) \\ 
\to(\exists \alpha\forall n B_n(\initSeg{\alpha}{n+1})\to R)\to \exists \alpha\forall n B_n(\initSeg{\alpha}{n+1}).\end{cases}\end{equation*}
which is precisely the dependent $J$-shift of \cite{EscOli(2012.0)}, and our realizer becomes
\begin{equation*}\PS_{(<)}^{\varepsilon,q}(s)=\ext{s}\at\lambda n\;. \; \varepsilon_{\initSeg{\alpha}{n}}(\lambda x\; . \; q(\PS_{(<)}^{\varepsilon,q}(\initSeg{\alpha}{n}\ast x)))\end{equation*}
which is isomorphic to the dependent product of selection functions of \cite{EscOli(2012.0)}. Setting $A_n(s,x):=B_n(x)$ as in Section \ref{sec-dependent-examples-MBR} we see that the dependent $J$-shift implies the non-dependent $J$-shift:
\begin{equation*}\forall n(\exists x B_{n}(x)\to R)\to \exists x B_{n}(x)) 
\to(\exists \alpha\forall n B_n(\alpha(n))\to R)\to \exists \alpha\forall n B_n(\alpha(n))\end{equation*}
and our realizer can be simplified to
\begin{equation*}\PS_{(<)}^{\varepsilon,q}(s)=\ext{s}\at\lambda n\;. \; \varepsilon_{n}(\lambda x\; . \; q(\PS_{(<)}^{\varepsilon,q}(\initSeg{\alpha}{n}\ast x)))\end{equation*}
which is now isomorphic to the non-dependent product of selection functions of \cite{EscOli(2012.0)}.

\subsubsection*{Summary}
\label{sec-dependent-examples-summary}

We have demonstrated that the many variants of bar recursion used to intepret choice principles are just instances of the same basic cominatorial idea. Theorem \ref{thm-main-DC} gives us a completely uniform framework in which to understand the variety of different realizers currently in the literature, and moreover in each case we are able to provide generalisations of these realizers to more complex orderings on $\N$. We summarise all this in the table below (here $m_1nu:=\least i\leq n(i\notin\dom(u))$ and $m_2nu=n$).\\

\

\begin{tabular}{c | c | c | c | l}$m$ & $\prec_u$ & $\lhd$ & $A_n(s,x)$ & \textbf{realizer} \\ \hline & & & \\  $m_1$ & $\emptyset$ & $\emptyset$ & $B_n(x)$ & Berardi-Bezem-Coquand functional \cite{BBC(1998.0)} \\ 
$m_2$ & $\emptyset$ & $\emptyset\;/<$ & $B_n(x)$ & simple modified bar recursion \cite{BergOli(2005.0)}  \\ 
$m_2$ & $\emptyset$ & $<$ & $B_n(s_{|s|-1},x)$ & dependent modified bar recursion \cite{BergOli(2005.0)}\\
$m_2$ & $\emptyset$ & $<$ & $B(s\ast x)$ & dependent modified bar recursion \cite{Seisenberger(2008.0)}\\
$m_1$ & $<$ & $\emptyset \;/<$ & $B_n(x)$ & simple product of selection functions \cite{EscOli(2010.1),EscOli(2012.0)} \\ 
$m_1$ & $<$ & $<$ & $B_n(s\ast x)$ & dependent product of selection functions \cite{EscOli(2010.1),EscOli(2012.0)}
\end{tabular}\\

\

However, these examples consitute only an extremely limited range of possibilities for $\Psi_{(m,\lhd,<)}$ based on very simple instantiations of its parameters. Theorem \ref{thm-main-DC} gives us far more than just a unifying perspective for existing realizers: it gives us a very general recipe for devising new realizers for choice principles that can be tailored to the situation at hand. Thus, rather than simply giving a fixed interpretation of choice and extracting programs relative to this realizer, we have an additional level of flexibility which should allow us to extract much more efficient and meaningful programs.

\section{Understanding the parametrised realizer}
\label{sec-semantics}

Having completed the main theoretic work of this paper, the purpose of this section is to give a somewhat informal graphical representation of the structure of our realizer and the proof of Theorem \ref{thm-main-DC}. In doing so we suggest a potential semantic interpretation of the realizer. 

Much work has been done in the last few decades on providing a computational interpretation of classical logic that can be understood on an intuitive level terms of \emph{learning} - the basic idea being that realizers of classical principles typically carry out some kind of `learning procedure' in order to construct an approximation to that principle. 

In \cite{BBC(1998.0)} a connection is suggested between realizers of negative translated formulas and winning strategies related to the Novikoff interpretation of classical formulas \cite{Coquand(1994.0),Novikoff(1943.0)}. In particular, an illuminating semantic interpretation of what we have called the Berardi-Bezem-Coquand functional is  given, in which the functional represents a strategy for building an approximation of a choice functional which wins against any continuous opponent. We attempt to extend this interpretation to our generalised realizer of choice, and argue that the our parameters can be giving a clear meaning in this context.

Let us first very briefly recall the basic ideas of \cite{BBC(1998.0),Coquand(1994.0)} (the reader is referred to these papers for a proper treatment). In the Novikoff calculus, the formula $\exists x\forall y A(x,y)$ is mapped to the propositional formula
\begin{equation*}\bigvee_x\bigwedge_y A(x,y).\end{equation*}
The truth of such a formula is debated by two players $\Eloise$ and $\Abelard$, who support truth and falsity respectively. First, $\Eloise$ selects some value $x_0$ for which she claims that $\bigwedge_y A(x_0,y)$ is true, and $\Abelard$ follows this by choosing some $y_0$ in an attempt to falsify $A(x_0,y_0)$. The formula as a whole is intuitionistically valid iff $\Eloise$ has a winning strategy regardless of any choices made by $\Abelard$.

Now, it is clear that such a correspondence does not work in the case of classical logic, since there are $\Sigma_2$-formulas classically true but for which there is no effective strategy for $\bigvee_x\bigwedge_y A(x,y)$. In this case, validity of $\bigvee_x\bigwedge_y A(x,y)$ is interpreted as the existence of some $x_0$ such that
\begin{equation*}A(x_0,y')\vee\bigvee_x\bigwedge_y A(x,y)\end{equation*}
is valid for all $y'$. The idea here is that $\Eloise$ picks a potential witness $x_0$, which is followed by an attempt at a counterexample $y_0$ from $\Abelard$, and the game becomes $A(x_0,y_0)\vee\bigvee_x\bigwedge_y A(x,y)$. In other words, either $A(x_0,y_0)$ is true, in which case $\Eloise$ wins, or it is false, and $\Eloise$ can backtrack and start again, this time using falsity of $A(x_0,y_0)$ as constructive information. In this way, $\Eloise$ is allowed to `learn' from $\Abelard$'s choices. Moreover, as demonstrated in \cite{BBC(1998.0)}, this notion of learning and backtracking is captured by the recursive functionals which realize the negative translation of the classical formula.

Let us now take this basic idea and consider how the truth of countable choice can be interpreted as a dialogue between $\Eloise$ and $\Abelard$ in which $\Eloise$ eventually wins. Countable choice can be written as the disjuction
\begin{equation*}\label{eqn-AC-dis}\exists n\forall x\neg A_n(x)\vee\exists f\forall n A_n(f(n)).\end{equation*}
We give this formula a rough interpretation along the lines of \cite{BBC(1998.0)} as follows. $\Eloise$ begins by attempting to realize the conclusion of $\AC$ with some default function $f_0=\lambda n.0$ whose `domain' of genuine constructive information is empty, and $\Abelard$ responds by selecting some point $n_0$ such that this attempt fails i.e. we cannot provide a realizer for $A_{n_0}(0)$. $\Eloise$ responds by now attempting to falsify the premise at some $m_0$, to which $\Abelard$ responds with a point $x_0$ and a realizer for $A_{m_0}(x_0)$. 

$\Eloise$ has now been given some constructive information by $\Abelard$, so she starts the whole process again, this time with the function $f_1:=f_0[m_0\mapsto x_0]$, which now has a domain of $\{m_0\}$. This time, $\Abelard$ picks a point $n_1$ at which $f_1$ fails. Either he picks $n_1=m_0$ in which case he loses since by his own admission $A_{m_0}(x_0)$ is true (see \cite{BBC(1998.0)} for a formal translation of this logic into Novikoff strategies), or he chooses $n_1\notin\{m_0\}$. Then $\Eloise$ responds with some $m_1\notin\{m_0\}$ which falsifies the premise of $\AC$, and again $\Abelard$ responds with some $x_1$ and a realizer for $A_{m_1}(x_1)$. $\Eloise$ now updates her approximation again to some $f_2$ which includes this information, and continues as before. Roughly speaking, such a strategy should eventually result in success for $\Eloise$ whenever $\Abelard$ is `continuous', because he will eventually be forced to pick $n_i$ in the domain of $f_i$.

We have been vague as to two details in this strategy: firstly how $\Eloise$ decides on which $m_i$ to choose in light of $\Abelard$'s original choice $n_i$, and secondly in how she chooses to update her approximation each time. In the first instance it is clear that if $\Abelard$ has a `good' choice for $n_i$ which is not in the domain $f_i$ then $\Eloise$ must also respond with $m_i$ not in the domain of $f_i$ if she has any chance of falsifying the premise of $\AC$. Then, when it comes to updating $f_i$ with the information $(m_i,x_i)$, she could either just add this directly to $f_i$ and define $f_{i+1}:=f_i[m_i\to x_i]$, or she could potentially erase some of the existing elements in the domain of $f_i$. 

The most natural choice is of course to pick $m_i=n_i$ each time and update directly without any erasing, and this is precisely the strategy given in \cite{BBC(1998.0)}. However, in the case of \emph{dependent choice} where the elements of the choice sequences are related in some way, it is crucial for $\Eloise$ to update her approximation in a manner which is coherent with the underlying dependency required for the choice sequence.

Let us now move on to our realizer for the negative translation of countable dependent choice (in the form of the $\oDC{\lhd}^{J^\ast_1}$-shift), and show how the parameters of the realizer correspond to the choices in $\Eloise$'s strategy we have just discussed. In Figure \ref{fig-a} below we give a (very informal) diagrammatic representation of the proof of Theorem \ref{thm-main-DC}. We retain all our notation, in particular we write $\bar A(u)$ to denote that $u$ is a partial realizer. For simplicity we now work in a concrete realizability setting with $R=\bot$ and $\varepsilon_{s,r}(p)=\pair{0,h(\phi_s(p))}$ where 
\begin{equation*}\phi\mr \forall s((\exists x A_{|s|}(s,x)\to\bot)\to\bot).\end{equation*}
and $\forall n,r,x (h\mr (\bot\to A_n(r,x)))$, so in particular $\varepsilon$ realizes the premise of $\oDC{\lhd}^{J^\ast_1}$.

%
%
%
%
%
%
%
%

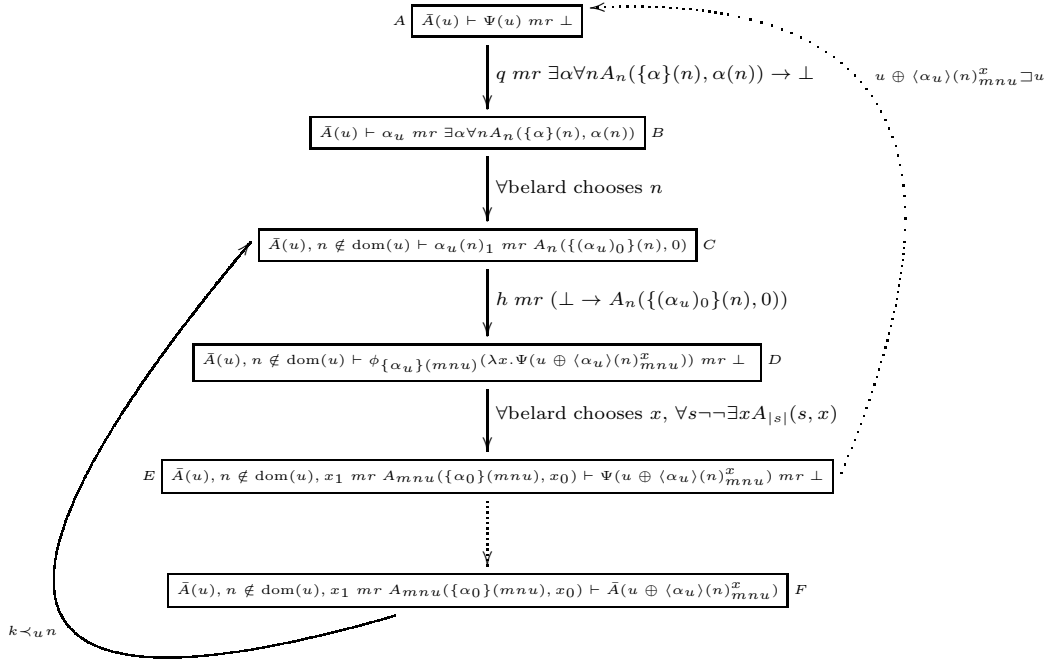
\begin{figure}[h]
\begin{center}
\[\tiny
\xymatrix{ & A \; \fbox{$\bar A(u)\vdash \Psi(u)\mr \bot$}\ar@/^8.0pc/@{<.}[dddd]!R^{ u\at \update{\initSegsss{\alpha_u}{n}}{mnu}{x}\sqsupset u}& \\ 
 & \fbox{$\bar A(u)\vdash \alpha_u\mr\exists\alpha\forall n A_n(\initSegss{\alpha}{n},\alpha(n))$} \; B\ar@{<-}[u]_{\mbox{\scriptsize $q\mr \exists\alpha\forall n A_n(\initSegss{\alpha}{n},\alpha(n))\to\bot$}} & \\ 
 & \fbox{$\bar A(u), n\notin\dom(u)\vdash \alpha_u(n)_1\mr A_n(\initSegss{(\alpha_u)_0}{n},0)$} \; C \ar@{<-}[u]_{\mbox{\scriptsize $\Abelard$ chooses $n$}} & \\
 & \fbox{$\bar A(u),n\notin\dom(u)\vdash \phi_{\initSegss{\alpha_u}{mnu}}(\lambda x.\Psi(u\at \update{\initSegsss{\alpha_u}{n}}{mnu}{x}))\mr\bot$ } \; D \ar@{<-}[u]_{\mbox{\scriptsize $h\mr (\bot\to A_n(\initSegss{(\alpha_u)_0}{n},0))$}} & \\
 & E \; \fbox{$\bar A(u),n\notin\dom(u),x_1\mr A_{mnu}(\initSegss{\alpha_0}{mnu},x_0)\vdash\Psi(u\at \update{\initSegsss{\alpha_u}{n}}{mnu}{x})\mr \bot$} \ar@{<-}[u]_{\mbox{\scriptsize $\Abelard$ chooses $x$, $\forall s\neg\neg \exists xA_{|s|}(s,x)$}} & \\
 & \fbox{$\bar A(u),n\notin\dom(u),x_1\mr A_{mnu}(\initSegss{\alpha_0}{mnu},x_0)\vdash \bar A(u\at \update{\initSegsss{\alpha_u}{n}}{mnu}{x})$} \; F \ar@{<.}[u]\ar@/^11.5pc/@{->}[uuu]!L^{k\prec_u n}  & 
}
\]
\end{center}
\caption{Diagrammatic illustration of Theorem \ref{thm-main-DC}}
\label{fig-a}
\end{figure}

Let us now run through the diagram, step by step. Roughly speaking, each box represents a stage in a game between $\Eloise$ and $\Abelard$, and an arrow represents a reverse implication in the proof of Theorem \ref{thm-main-DC}. We start at $A$ with the assumtion that $\Eloise$ has already computed a partial realizer $u$ that is correct wherever it's defined. At step $B$, $\Eloise$ plays an approximation $(\alpha_u)_0:=u_0\at 0$ to a choice sequence with a corresponding sequence $(\alpha_u)_1$ of realizers, and in response, $\Abelard$ selects some $n$ and challenges $\Eloise$ to realise $A_n(\initSegss{(\alpha_u)_0}{n},(\alpha_u)_0)$. If $n\in\dom(u)$ then $\Eloise$ wins, so we assume the contrary and move on to step $C$. $\Eloise$ responds to the challenge in the next step $D$ by claiming that the premise of choice principle is false at point $mnu$. $\Abelard$ is now forced to produce $x$ such that $x_1\mr A_{mnu}(\initSegss{(\alpha_u)_0}{mnu},x_0)$ - if he fails then $\Eloise$ wins, and if he succeeds then $\Eloise$ is given constructive information and takes this as an assumption. 

We now come to the subtle part: the manner in which $\Eloise$ updates $(\alpha_u)_0$ to reflect this new information and repeat the loop. From a proof theoretic perspective, $E$ is implied by $A$ and $F$ by the cut rule. We can interpret this semantically as follows: $\Eloise$ forms the updated function $(\alpha_u)_0[mnu\mapsto x_0]$, but states that if in future $\Abelard$ queries this realizer for any $k\prec_u n$ with $k\notin\dom(u)$, she will ignore any subsequent information received and revert to stage $C$ as if $\Abelard$ had chosen $k$ instead of $n$. This is reflected in the definition of $\Psi$, since it makes a recursive call on the partial function $u\at \update{\initSegsss{\alpha_u}{n}}{mnu}{x}$, but the fact that this is a partial realizer for $k\prec_u n$ relies on nested recursive calls of the form $\Psi(u\at\update{\initSegsss{\alpha}{k}}{mku}{y})$, which in particular forgets the value of $x$ at $mnu$.  

Therefore at stage $E$ there are two possibilities. Either at some point in the future $\Abelard$ does query $k\prec_u n$ for $k\notin\dom(u)$, in which case all subsequent information is deemed irrelevant and we make an auxiliary $\prec_u$-recursive loop back to $C$, or $\Abelard$ never queries $k\prec_u n$, in which case we can treat $u\at \update{\initSegsss{\alpha_u}{n}}{mnu}{x}$ as a partial realizer and make a backward recursive loop back to stage $A$. In this way, the game corresponding to dependent choice can be seen as a path through Figure \ref{fig-a}. By the combination of $\prec_u$-induction and backward induction, using the fact that $\prec_u$ is well-founded and $\Abelard$'s choice of $n$ at step $B$ is based on a continuous strategy, the proof of Theorem \ref{thm-main-DC} essentially says that there is no infinite path starting from $\emptyset$, and therefore $\Eloise$ has a winning strategy in the game as a whole.

Of course, we have not made our proposed link between Theorem \ref{thm-main-DC} and the world of Novikoff games precise! But at the very least we hope to have provided some insight into how our somewhat syntactic construction in the previous sections could be viewed from an intuitive, semantic perspective.

Let us now think about our concrete examples. For the Berardi-Bezem-Coquand functional we have $\initSegsss{\alpha_u}{n}=\emptyset$ and so the `forgetful' subloop is completely avoided. In this case $\Eloise$ simply responds to $\Abelard$'s choice of $n$ directly, and so when $\Abelard$ plays $(\pair{n_i,x_i})=\pair{1,x_0},\pair{4,x_1},\pair{3,x_2},\ldots$ we get the following sequence of updates 
\begin{equation*}[0,0,0,0,0]\mapsto [0,x_0,0,0,0]\mapsto [0,x_0,0,0,x_1]\mapsto [0,x_0,0,x_2,x_1]\mapsto\ldots \end{equation*}
This semantic interpretation of the BBC-functional is of course completely analogous to the one given in \cite{BBC(1998.0)}. For modified bar recursion, whenever $\Abelard$ makes a sensible choice $n\notin\dom(u)$, $\Eloise$ switches instead the least element not already in the domain. Again we avoid the forgetful subloop, but this time $\Eloise$'s updates are done in sequence - the same three choices from $\Abelard$ results in the following response:
\begin{equation*}[0,0,0,0,0]\mapsto [x_0,0,0,0,0]\mapsto [x_0,x_1,0,0,0]\mapsto [x_0,x_1,x_2,0,0]\mapsto\ldots \end{equation*}
Finally, for the product of selection functions, $\Eloise$ never changes $\Abelard$'s choice of $n$, but now assumes a policy of forgetting everthing above the point being updated. This time the result is
\begin{equation*}[0,0,0,0,0]\mapsto [0,x_0,0,0,0]\mapsto [0,x_0,0,0,x_1]\mapsto [0,x_0,0,x_2,0]\mapsto\ldots \end{equation*}
the point $x_1$ being forgotten as a result of $\Abelard$ choosing $3<4$ as his third move.

The strategy related to the BBC-functional fails for dependent choice because it is assumed that $\Abelard$ always picks $x_i$ such that $A_m(\initSeg{u\at 0}{m},x_i)$ holds, and so e.g. we have $A_4([0,x_0,0,0],x_1)$. However, if ever in the future he gives $\Eloise$ some new information $x_{i+j}$ for point $k<m$ in the sequence, then $x_i$ is no longer valid - for instance we would require $A_4([0,x_0,0,x_2],x_1)$ to hold, which is not necessarily true.

Modified bar recursion and the product of selection functions represent two methods of overcoming this. For the former, $\Eloise$ ensures that $\Abelard$ always gives her constructive information in sequence, while for the latter she is happy to take information for any point $\Abelard$ chooses, but whenever updating she erases everthing above, relying on future moves to regain this information.

\section{Concluding remarks}
\label{sec-concluding}

There are several directions in which the work presented here could be developed. We have already demonstrated that Theorem \ref{thm-main-DC} provides a uniform soundness proof through which most of the existing variants of backward recursion used to interpret dependent choice can be derived. The most obvious next step would be to explore the use of new variants of bar recursion which arise from giving more interesting values to the parameters of $\Psi_{(m,\prec,\lhd)}$.

Take, for example, the following classical statement:
\begin{quote}For any function $f\colon\Bool^\ast\to \N$, there exists a function $g\colon\Bool^\ast\to\N$ such that for any branch $s\colon\Bool^\ast$ we have $f(g(s))\leq f(t)$ whenever $t$ has $s$ as a prefix. \end{quote}
This is a direct consequence of the axiom of countable choice, and so any instantiations of $(m,\prec)$ satisfying the conditions of Theorem \ref{thm-main-AC} will yield a realizer capable of building an approximation to the choice function $g$. However, it would seem most sensible to choose the parameters so that the recursion is carried out over the natural tree structure $s\lhd t$ iff $s$ is a prefix of $t$ (relative to some encoding of $\Bool^\ast$ into $\N$) - using, for example, the variant of modified bar recursion sketched out the beginning of Section \ref{sec-dependent-examples-MBR}. Then if an approximation to $g$ is defined at $s$ then this information could be used to extend the approximation to extensions $t$ of $s$.

It would be interesting to examine in general how such realizers, tailored to the situation at hand, compare to those built from the existing forms of bar recursion. One would expect an advantage in terms of both algorithmic efficiency and the syntactic expressiveness of the extracted program.

Another interesting application of our parametrised form of bar recursion would be to extend the work of Escard{\'o} and Oliva on Nash-equilibria of unbounded games to the transfinite case. In \cite{EscOli(2011.0),EscOli(2012.0)} it is shown that the infinite products of selection functions corresponding to the functional $\PS_{(<)}$ defined in Section \ref{sec-dependent-examples-IPS} computes optimal strategies in infinite sequential games over the natural numbers. However, $\PS_{(\lhd)}$ is well-defined for any computable well ordering on $\N$, and in particular it is not too difficult to show that in these cases Spector's equations:
\begin{equation*}\begin{aligned}\alpha(n)&=\varepsilon_{\initSegss{\alpha}{n}}(p_n)\\ p_n(\varepsilon_{\initSegss{\alpha}{n}}(p_n))&=q(\alpha)\end{aligned}\end{equation*}
can be solved in $\varepsilon$ and $q$ for arbitrary $\lhd$ by setting $\alpha=\PS^{\varepsilon,q}(\emptyset)$ and $p_n(x):=q(\PS^{\varepsilon,q}(\initSegss{\alpha}{n}\star x))$, entirely analogously to the normal ordering on $\N$. However, it would be useful to formalise this properly and to investigate whether there are any interesting applications of higher-type \emph{transfinite} games.

Finally, to the author's knowledge all known realizability interpretations of choice are restricted to countable or dependent choice principles, and it is not known how to extend these to choice over function spaces, for example:
\begin{equation*}\forall f^{\N\to\N}\exists x^\rho B_f(x)\to\exists F^{(\N\to\N)\to\rho}\forall f B_f(F(f)).\end{equation*}
One possibility would be to build an approximation to $F$ over some countable basis $(c_i)_{i\in\N}$ for the space $\N\to\N$. However, in this case updates to the approximation must be made in a coherent way, since the values of two distinct elements $c_1$ and $c_2$ must be compatible with their intersection. Therefore the ideas behind our parametrised realizer could be helpful here, the aim being to assign an ordering to the basis and make sure that the updates respect this ordering. However, in the author's opinion it is likely that some additional ingenuity would be required here, as the challege posed by giving a computational interpretation non-countable choice seems to be a significant one.

\

\noindent\textbf{Funding.} This work was supported by a LabEx CARMIN postdoctoral research fellowship, and also the Austrian Science Fund (FWF) project ``Automated Complexity Analysis via Transformations'' (project number P25781).

\appendix

\section{Appendix}
\label{sec-app}

\begin{proof}[Proof of Proposition \ref{prop-main-DC}]We define $\Psi_{(\lhd,\prec,m)}^{\varepsilon,q}(u):=\BR^{\psi^{\varepsilon,q}}(u)$ where
\begin{equation*}\psi^{\varepsilon,q}_u(f^{\N\times \overline{(\rho\times\sigma)}^\N\to\tau}):=_\tau q(\alpha^{\varepsilon,f}_u)\end{equation*}
where $\alpha^{\varepsilon,f}_u\colon (\rho\times\sigma)^\N$ is recusively defined over $\prec'_u$ by
\begin{equation*}\alpha^{\varepsilon,f}_u(n):=\begin{cases}u(n) & \mbox{if $n\in\dom(u)$}\\
\varepsilon_{\initSegss{\alpha}{mnu},\initSegss{\alpha_0}{n}}(\lambda x\; . \; f(mnu,u\oplus \update{\initSegsss{\alpha_u}{n}}{mnu}{x})) & \mbox{otherwise}\end{cases}\end{equation*}
where $\initSegsss{\alpha_u}{n}:=\lambda k.(\alpha_u^{\varepsilon,f}(k)\mbox{ if $k\prec_u n$})$. By condition (\ref{item-DC-conda}) of the proposition, $\alpha_u^{\varepsilon,f}$ is well-founded. Now, unwinding definitions, we obtain
\begin{equation*}\begin{aligned}\Psi^{\varepsilon,q}(u)&=\psi_u(\underbrace{\lambda n,v\; . \; \Psi^{\varepsilon,q}(u\at v)\mbox{ if $n\in\dom(v)\backslash\dom(u)$}}_{f})\\
&=q(\alpha_u)\\
&=q(u\at\lambda n\; . \;\varepsilon_{\initSegss{\alpha}{mnu},\initSegss{\alpha_0}{n}}(\lambda x\; . \; f(mnu,u\oplus \update{\initSegsss{\alpha_u}{n}}{mnu}{x})))\\
&\stackrel{(\ast)}{=}q(u\at\lambda n\; . \;\varepsilon_{\initSegss{\alpha}{mnu},\initSegss{\alpha_0}{n}}(\lambda x\; . \; \Psi(u\oplus \update{\initSegsss{\alpha_u}{n}}{mnu}{x}))) \end{aligned}\end{equation*}
where step $(\ast)$ follows from
\begin{equation*}\begin{aligned}f(mnu,u\oplus \update{\initSegsss{\alpha_u}{n}}{mnu}{x})&=\Psi(u\oplus(u\oplus \update{\initSegsss{\alpha_u}{n}}{mnu}{x}))\mbox{ if $mnu\in\dom(u\oplus \update{\initSegsss{\alpha_u}{n}}{mnu}{x})\backslash\dom(u)$}\\
&=\Psi(u\oplus \update{\initSegsss{\alpha_u}{n}}{mnu}{x})\mbox{ if $mnu\notin\dom(u)$}\\
&=\Psi(u\oplus \update{\initSegsss{\alpha_u}{n}}{mnu}{x})\end{aligned}\end{equation*}
the last step following from condition (\ref{item-DC-condb}).\end{proof}

\bibliographystyle{plain}
\bibliography{/home/thomas/Documents/tp}

\end{document}